\documentclass[english]{article}

\usepackage[margin=1.0in]{geometry}

\usepackage{placeins}
\usepackage[latin9]{inputenc}
\usepackage{amsthm}
\usepackage{amsmath}
\usepackage{amssymb}
\usepackage{graphicx}
\usepackage{esint}
\usepackage{hyperref}
\usepackage{paralist}
\usepackage{array}
\usepackage{booktabs}
\usepackage{xcolor} 
\usepackage{epstopdf}
\usepackage{comment}
\usepackage{booktabs}

\graphicspath{{Figures/}}

\allowdisplaybreaks
\def\R{{\mathbb R}}% real numbers
\def\RR{{\mathbb{R}}}

\def\le{\leqslant}% lessoreqal
\def\ge{\geqslant}%greaterorequal

\newcommand{\wh}{\widehat}
\newcommand{\wt}{\widetilde}
\newcommand{\ud}{\,\mathrm{d}}
\newcommand{\uD}{\,\mathrm{D}}

\newcommand{\Or}{\mathcal{O}}
\newcommand{\bd}{\boldsymbol}

\theoremstyle{plain}
\newtheorem{theorem}{Theorem}[section]

\theoremstyle{definition}

\newtheorem*{remark*}{Remark}

\newcommand{\mc}[1]{\mathcal{#1}}

\newcommand{\abs}[1]{\lvert#1\rvert}
\newcommand{\Abs}[1]{\left\lvert#1\right\rvert}

\newcommand{\bra}[1]{\langle#1\rvert}

\newcommand{\ket}[1]{\lvert#1\rangle}

\newcommand{\wb}[1]{\overline{#1}}

\newcommand{\Var}{\text{Var}}

\DeclareMathOperator{\tr}{Tr}

\newcommand{\jl}[1]{\textbf{\textcolor{red}{JL: #1}}}

\begin{document}

\title{Continuum limit and preconditioned Langevin sampling of the path integral molecular dynamics}

\author{Jianfeng Lu \footnote{Department of Mathematics, Department of Physics and
 Department of Chemistry, Duke University, Durham NC 27708, USA} \  \
	            Yulong Lu \footnote{Department of Mathematics, Duke University, Durham NC 27708, USA} \   \
	            Zhennan Zhou\footnote{Beijing International Center of
  Mathematical Research, Peking University, Beijing 100871,  P.R. China} 
}

%\author{Jianfeng Lu} \email{jianfeng@math.duke.edu}
%\affiliation{Department of Mathematics, Department of Physics and
%  Department of Chemistry, Duke University, Durham NC 27708, USA}
%\author{Zhennan Zhou} \affiliation{Beijing International Center of
%  Mathematical Research, Peking University, Beijing 100871,
%  P.R. China}

\maketitle

\begin{abstract}
  We investigate the continuum limit that the number of beads goes to
  infinity in the ring polymer representation of thermal
  averages. Studying the continuum limit of the trajectory sampling
  equation sheds light on possible preconditioning techniques for
  sampling ring polymer configurations with large number of beads. We
  propose two preconditioned Langevin sampling dynamics, which are
  shown to have improved stability and sampling accuracy. We 
  present a careful mode analysis of the preconditioned dynamics and show their
  connections to the normal mode, the staging coordinate and the Matsubara mode representation for
  ring polymers. In the case where the potential is quadratic, we show that the continuum limit of the preconditioned
 mass modified Langevin dynamics converges to its equilibrium exponentially fast, which suggests that the finite dimensional counterpart  has a dimension-independent convergence rate.
 {In addition, the preconditioning techniques can be naturally applied to
  the multi-level quantum systems in the nonadiabatic regime, which are compatible 
  with various numerical approaches.}
\end{abstract}

\section{Introduction}

Simulating complex chemical systems including quantum effects has been one of the central research subjects in  computational chemistry and physics. For the thermal average calculation, the ring polymer representation,
based on the imaginary time path integral, has been a popular approach
to map a quantum particle in thermal equilibrium to a fictitious
classical necklace of beads on the phase space
\cite{Feynman:72}. As the number of beads in the ring polymer goes to infinity, the  representation is asymptotically exact, and thus it provides an approximate formulation for numerical simulations when the bead number is large.  Based on the
ring polymer representation, there are mainly two types of sampling techniques to calculate the quantum statistical averages, namely, path integral Monte Carlo (PIMC)
\cite{ChandlerWolynes:81, BerneThirumalai:86} and path integral
molecular dynamics (PIMD) \cite{MarklandManolopoulos:08,
  CeriottiParrinelloMarklandManolopoulos:10}. Also, the ring
polymer representation has been used in the dynamics simulations,
such as, the centroid molecular dynamics \cite{CaoVoth:94a,
  CaoVoth:94b, JangVoth:99}, the ring polymer molecular dynamics
\cite{CraigManolopoulos:04, HabershonManolopoulosMarklandMiller:13},
Matsubara dynamics \cite{HeleWillattMuoloAlthorpe:15a,
  HeleWillattMuoloAlthorpe:15b}, and path integral Liouville dynamics
\cite{Liu2014, LiuZhang2016}. In recent years, the ring polymer representation based methods, like path integral molecular dynamics, have been extended to the multi-level systems when the non-adiabatic effects are not negligible \cite{MeyerMiller:79, StockThoss:97,StockThoss:05,LZPIMD1,LZPIMD2,LL18}.

When the number of beads is finite, the ring polymer representation
introduces quantifiable asymptotic error to thermal average
calculations, which is manifested in the bias of the numerical
simulations. When the physical space of the quantum particle is high
dimensional, {when the inverse temperature is large,} or when the
potential landscape is complicated, a large number of beads are needed
to decrease the model error. However, as the number of beads
increases, the spring potential for the ring polymer becomes more
stiff, which brings additional challenges for numerical sampling.
More precisely, the highest frequency of the normal modes of the ring
polymer increases as as the number of beads increases, which restricts
the time steps for the numerical integration of the sample trajectory.

%\jl{The following paragraph needs to be substantially rewritten; we
%  shall first introduce methods that are much older in the literature;
%  and then talk about some recent works. I also would not say it is a
%  ``growing interest'' if the literature dates back to 80s.} 

{Improved numerical approaches based on the ring polymer
  representation have received much attention.  Efficient sampling
  techniques are introduced by decomposing the potential function or
  by exploring thermostatting methods, see e.g.,
  \cite{FD84,CBP09,MarklandManolopoulos:08,CeriottiParrinelloMarklandManolopoulos:10,CMP11,CPMM10,LiuLiLiu:16,KBM19}. It
  is also proposed to adjust the artificial masses for the momentum
  variables in the path integral representation to improve the
  numerical efficiency in the framework of normal modes or staging
  coordinates, see e.g.,
  \cite{TBMK93,CM93,CM04,MarklandManolopoulos:08,LiuLiLiu:16}.  More
  recently, there are also some attempts by making use of the
  Matsubata modes, which can be viewed as the finite dimensional
  approximation to the limit that the number of beads goes to infinity
  \cite{Matsubara,HeleWillattMuoloAlthorpe:15a,HeleWillattMuoloAlthorpe:15b}.
  We also remark that previous works mostly focus on finite number of
  modes without considering the infinite dimensional limit of the
  Gibbs measure in the path integral representation, and thus the
  understanding of the improving techniques above are limited to the
  finite dimensional cases.  }

 In this work, we focus on the thermal averages defined in the following form
\[
\langle\wh{A}\rangle = \frac{1}{\mathcal{Z}} \tr[e^{-\beta \wh{H}} \wh{A}],
\]
where $\wh{H}$ is a quantum Hamiltonian operator, $\wh{A}$ is an
observable, $\beta$ is the inverse temperature, and
$\mathcal{Z} = \tr[e^{-\beta\wh{H}}]$ is the partition function.

Formally speaking, when the number of beads goes to infinity, the ring
polymer converges to a (closed) Brownian path in the configuration
space such that the two ends agree. This continuum limit has been
studied in math literature, where the closed Brownian paths are
referred to as Brownian loops (see e.g., \cite{Hairer16}, where the
motion of the random loops on a Riemannian manifold is considered). In
our current context, we define the space of all loops on the Euclidean
space $\mathbb R^d$, denoted by $\mathcal L \mathbb R^d$, as
\[
\mathcal L \mathbb R^d := \{ \mathfrak q: \, [0, \, \beta] \rightarrow \mathbb R^d, \, \mathfrak q(0)=\mathfrak q(\beta) \}.
\]
The energy of the loop $\mathfrak q \in \mathcal L \mathbb R^d$ is given by
\begin{equation}\label{eq:energy}
E(\mathfrak q)=  \int_0^{\beta} \left[ \frac 1 2 |\partial_{ \tau} \mathfrak  q(\tau)|^2 + V(\mathfrak q(\tau)) \right] \ud \tau.
\end{equation}
We emphasize that the inverse temperature appears in the integration
limit in the definition of the loop energy \eqref{eq:energy}. On the
loop space $\mathcal L \mathbb R^d$, we formally  define the (infinite dimensional) Gibbs probability measure $\pi$ on the loop
space by 
\[
\pi(\mathrm{d} \mathfrak{q}) \propto \exp(-E(\mathfrak  q)) \uD[\mathfrak q],
\]
where $\uD[\mathfrak q]$ is a  flat reference measure on the loop space. 
Of course this picture is not rigorous since there is no Lebesgue measure in a infinite dimensional space. Nevertheless, this is rigorous at the level of finite dimensional distribution and a proper definition of $\pi$ is given in Section \ref{sec:pi_cont} by taking a Gaussian measure as the reference measure. With the measure $\pi$ formally defined as above, we have   
\begin{equation}\label{eq:pathintegral}
\begin{aligned}
  \langle\wh{A}\rangle & = \mathbb{E}_{\pi(\mathrm{d}
    \mathfrak{q})}\biggl[
  \frac{1}{\beta} \int_0^{\beta} A(\mathfrak{q}(\tau)) \ud \tau \biggr] \\
  & = \frac{1}{\beta \mathcal{Z}} \int \int_0^{\beta} A(\mathfrak{q}(\tau)) \ud \tau \; e^{-E(\mathfrak{q})} \uD[\mathfrak q] \\
  & = \frac{1}{\beta \mathcal{Z}} \int \int_0^{\beta}
  A(\mathfrak{q}(\tau)) \ud \tau \; \\
& \quad \quad \times e^{-\int_0^{\beta}  \frac{1}{2}
    \abs{\partial_{\tilde \tau} \mathfrak{q}(\tilde \tau)}^2 + V(\mathfrak{q}(\tilde \tau))\ud
    \tilde \tau} \uD[\mathfrak q],
\end{aligned}
\end{equation}
which is the Euclidean path integral representation of the quantum
thermal average. We emphasize again that the partition function $\mathcal{Z}$ is not properly defined here since $\uD[\mathfrak q]$ is not well-defined; see \eqref{Z_N} for the definition in the finite dimensional case.

To sample the measure $\pi$ on the path space, one can generalize the
conventional overdamped Langevin equation to the following stochastic
partial differential equation (SPDE) 
\[
\frac{\ud \mathfrak  q}{\ud t}  = -\frac{ \delta E(\mathfrak q) }{\delta \mathfrak q}  + \xi=\partial_{\tau\tau}\mathfrak  q -\nabla V(\mathfrak q) + \xi ,
\]
which produces ergodic path of loops with respect to $\pi$. Here, we use $\tau$ to denote the parametrization variable of the path and denote by $\xi$
 the space-time white noise. This SPDE can be viewed as a continuum
limit of the overdamped Langevin equation of ring polymers with finite number of beads.

While to the best of our knowledge, the perspective of loop sampling
has not been much explored in the context of PIMD, it is closely
related to the question of sampling diffusion bridges (Brownian paths
with fixed boundary conditions) in applied mathematics and statistics
literature, e.g., \cite{HSVW05,HSV07}. In particular, efficient
numerical algorithms for sampling the diffusion bridges have been
extensively studied in the past decade, see e.g., \cite{BRSV08,
  BPSS11, HSV14, Eberle14}.  In particular preconditioning of the
infinite-dimensional SPDE have been explored to improve the sampling
efficiency, see e.g. \cite{cotter2013mcmc, beskos2017geometric,lan2019adaptive, BPSS11,ottobre2016function}.
%\ylcomment{Here I added some references on infinite dimensional sampling or preconditioning. Please check!}

%    % few years, where the finite dimensional approximations are based on the Wiener-L\'evy expansion or the Karhunen-Lo\`eve expansion.
% In \cite{HSV14},  a dimension-independent Wasserstein spectral gap is shown for preconditioned Crank-Nicolson (pCN) algorithms for a large class of measures of random paths, and this
% Wasserstein spectral gap implies an $L^2$-spectral gap. Similar analysis for the MALA chain method
% with semi-implicit Euler proposals and the Metropolis-Hastings chain method
% with Ornstein-Uhlenbeck proposals are shown in \cite{Eberle14}. And in \cite{BRSV08}, the authors study MCMC sampling methods  with or without preconditioning, which are reversible with respect to the target diffusion bridge measure.  In another recent work \cite{BPSS11}, the authors develop a generalized Hybrid Monte Carlo algorithm arising as finite-dimensional approximations of random path measures on the infinite-dimensional phase space. In particular, they have introduced an numerical algorithm which can be viewed as a preconditioning approach to the Hamiltonian flow before taking the continuum limit. 

Motivated by the similarity between PIMD and diffusion bridge
sampling, in this work, we study the continuum limit of PIMD when the
number of bead goes to infinity. Based on the analytic understanding,
we propose two preconditioned Langevin sampling dynamics, both
motivated by recent techniques proposed in the context of diffusion
bridge sampling.  Our first preconditioning approach is based on the
idea of applying the covariance operator as in
\cite{HSV07,BRSV08,HSV14,ottobre2016function}. After applying the covariance operator
properly in the context of PIMD, the frequencies of the normal modes
of the preconditioned Langevin dynamics have a uniform upper bound
with respect to the bead number. In particular, the more important
modes are mapped to the ones closer to the upper bound such that those
modes are favored in sampling.  Secondly, inspired by the treatments
on the phase space as in \cite{BPSS11} and \cite{ottobre2016function}, we propose a mass-modified
Langevin dynamics, where the frequencies of all normal modes are
adjusted to $1$ in the harmonic case.  Such mass-modified Langevin
dynamics naturally connects to their continuum limit, the SPDEs on the
phase space, and can be reformulated to facilitate constructing
efficient numerical algorithms.  In addition, we prove in the harmonic case that 
the law of the preconditioned SPDEs converges to the equilibrium measure exponentially with an explicit rate.

{We emphasize that studying the
  continuum limit of the sampling measure and preconditioning
  techniques for the sampling equations in the infinite dimensional
  limit naturally give their finite dimensional counterparts. The
  subtle paradigm shift lies in that while most previous works study
  the preconditioning in the large but finite dimensional cases, we
  make use of the continuum limit to guide us in designing
  preconditioned sampling dynamics. Thus, we expect the proposed
  sampling equations to exhibit uniformly satisfying performances with
  respect to the number of beads.}  

{ For finite dimensional systems, the techniques boil down to
  regularizing the spring potential with a constable multiple of
  {the quadratic potential} and applying its inverse to
  precondition the sampling dynamics. The proposed preconditioning
  strategy can be easily extended to the multi-level quantum systems,
  where the nonadianatic interplay between the electronic energy
  levels is present.  }
%\jl{moved here; seems a better  position; some polishing is also needed}

% Besides, the connection between such
% reformulations and the normal mode representation, and the Matsubara
% mode representation are established to better illustrate the physical
% intuitions behind the mathematical treatment.

To sum up, we propose two preconditioned Langevin dynamics, which are
shown to have superior numerical performances in stability and
sampling accuracy.  The rest of the paper is outlined as follows. We
study the continuum limit of the ring polymer approximation in
position variables in Section \ref{sec:pi_cont}, and the continuum
limit of the overdamped Langevin dynamics and its preconditioned
version are introduced in Section \ref{sec:overdamp}. In Section
\ref{sec:underdamp}, two preconditioned underdamped Langevin dynamics
are proposed with different choices of mass matrices of the auxiliary
momentum variables, where the latter one is shown to connect with the
continuum limit of the Gibbs measure on the phase space. The two
preconditioned Langevin dynamics are further analyzed and compared
with the normal mode representation, {the staging coordinate
  representation,} and the Matsubara representation in Section
\ref{sec:com} and Section \ref{sec:mats}. The convergence to the equilibrium  of the preconditioned SPDEs in the harmonic case is proved in Section \ref{sec:conv}. {In Section \ref{sec:multi-level}, we discuss the preconditioning techniques for the multi-level quantum systems.} Finally, in Section \ref{sec:num}, we discuss the minor
modifications in numerical implementations, and provide extensive
numerical tests.

\section{Continuum limit of the ring polymer representation} \label{sec:theory}

%\zz{we have some "continuous limit"s and some "continuum limit"s, shall we stick to only one?} \jl{continuum limit is the correct one; please make sure we only use that.}

In this section, we investigate the continuum limit of the ring
polymer representation, as the number of beads goes to infinity.
We also study the continuum limits of the overdamped
and underdamped Langevin sampling schemes for the ring polymer
configurations. The preconditioning methods will be introduced to
overcome the stiffness of the dynamics when the number of beads is
large. {As we will consider the continuum limit, to distinguish,
  we will use regular font (such as $q$) for position configuration in
  $\RR^d$ for one bead, bold fonts such as $\bd{q}$ for the ring
  polymer with finite beads, and Fraktur fonts such as $\mathfrak{q}$
  for the continuous path.}

%\zz{shall we add something here? to highlight the motivation again?}

\subsection{Path integral for thermal average} \label{sec:pi_cont}

Let us consider the quantum Hamiltonian
\[ 
\wh H = \wh T+ \wh V = \frac{\wh p^2}{2}  + V(\wh q),
\]
where the particle mass is fixed as $1$ for simplicity of notation. We
consider the thermal equilibrium average, given by
\begin{equation}\label{eq:aveA} 
\langle\wh{A}\rangle = \frac{1}{\mathcal Z} \tr
[e^{-\beta \wh H} \wh A ],
\end{equation} 
for an 
observable $\wh A$, where $\beta = \frac{1}{k_B
T}$ with $k_B$ the Boltzmann constant and $T$ the absolute
temperature.

The ring polymer representation approximates  the thermal average as (up to a
normalization) an average with respect to the classical Gibbs
distribution for ring polymers on the configurational space as
%\jl{check the approximation order}
\begin{equation}\label{eq:ensembleavgA2} 
  \langle \wh{A} \rangle = 
\frac{1}{\mathcal Z_N} \int_{\mathbb R^{dN}} 
 \frac{1}{N} \sum_{i=1}^N A(q_i) e^{- S_N(\bd{q})} \ud \bd q 
 + \Or(N^{-2}),
\end{equation}
where $N$ is the number of beads, the action is given by
\begin{equation}
S_N(\bd{q}) = \frac{\beta}{N} \sum_{i=1}^N \Bigl[  \frac{(q_i - q_{i+1})^2}{2 (\beta/N)^2}  +  V(q_i) \Bigr],
\end{equation}
and $\mathcal Z_N$ is a normalization constant %\jl{why using $\mathcal Z_N'$ instead of $\mathcal Z_N$?} 
\begin{equation}\label{Z_N}
  \mathcal Z_N = \int_{\mathbb R^{dN}} e^{- S_N(\bd{q})} \ud \bd q, 
\end{equation}
so that 
\begin{equation*}
  \ud \mu_N(\bd{q}) = \frac{1}{\mathcal Z_N} \exp(-  S_N ( {\bd{q}}))\ud \bd{q}
\end{equation*}
is a probability measure on the ring-polymer configurations.

Sufficiently large number of beads $N$ is needed to reduce the
asymptotic error in the ring polymer representation. However, a large
$N$ increases the stiffness associated to the ``spring potential'' in
the action:
$\frac{1}{2} \bigl(\frac{N}{\beta}\bigr)^2 (q_i - q_{i+1})^2$. In
particular, as illustrated in Figure~\ref{fig:Nring}, as the number of
beads increases, a typical ring polymer configuration exhibits small
scale oscillations. In the limit $N \rightarrow \infty$, it converges
to random loops with local regularity similar to a Brownian motion.

\begin{figure}[ht]
\begin{centering}
\includegraphics[scale=0.55]{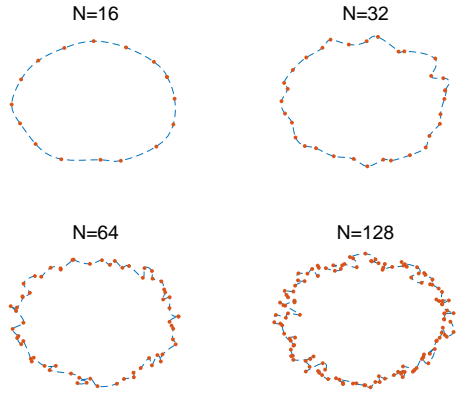} \\
\caption{Illustrative plots of ring polymers as the number of beads
  increases, made by initializing the beads evenly along a circle and
  numerically evolve the beads for a tiny time period using the
  Langevin dynamics \eqref{eq:Lang}.}\label{fig:Nring}
\end{centering}
\end{figure}

The ring polymer configuration, in the limit $N \to \infty$, can be
represented as a path $\mathfrak q(\tau):$ $[0,\beta] \mapsto \R^{d}$
with periodic boundary condition $\mathfrak q(0)=\mathfrak q(\beta)$.
From the ring polymer configuration with $N$ beads, we may construct a
corresponding path by setting $\mathfrak q(\tau_j) = q_j$ for
$\tau_j=\beta \tfrac{j}{N}$ and linearly interpolating in between. As
$N \to \infty$, we obtain
\begin{equation*}
  S_N(\bd{q}) \to \int_0^\beta \left[ \frac {|\dot {\mathfrak  q} (\tau)|^2}{2}+ V(\mathfrak q(\tau)) \right]d\tau
\end{equation*}
and 
\begin{equation*}
  \frac{1}{N}\sum_{i=1}^N A(q_i) \to \frac{1}{\beta} \int_0^{\beta} A(\mathfrak{q}(\tau)) d \tau.
\end{equation*}
Thus the path integral representation \eqref{eq:pathintegral} is
formally justified in the limit.

To make the characterization of the limit $N \to \infty$ more
rigorous, note that we can interpret $\mu_N$ as
\begin{equation*}
  \ud \mu_N(\bd{q}) \propto \exp\Bigl( - \frac{\beta}{N} \sum_{i=1}^N V(q_i) \Bigr) \ud \nu_N(\bd{q})
\end{equation*}
where $\nu_N$ is the finite dimensional Gaussian measure given by the harmonic spring potential part of $S_N$
\begin{equation}
\begin{aligned}
\ud \nu_N(\bd{q}) & \propto \exp\Bigl( - \frac{1}{2}  \frac{\beta}{N} \sum_{i=1}^N \frac{(q_i - q_{i+1})^2}{(\beta/N)^2}  \Bigr) \ud \bd{q} \\
& = \exp\Bigl( - \frac{1}{2}  \frac{\beta}{N} \bd{q} \cdot L \bd{q} \Bigr) \ud \bd{q}.
\end{aligned}
\end{equation}
Here the last equality gives the definition of $L$ as a $dN \times dN$
dimensional positive semi-definite matrix, with its lowest eigenvalue
being $0$, {associated with the eigenvector that all $q_i$'s are
  equal, so that the spring potential takes value $0$}.

To move the spectrum of $L$ away from $0$, we introduce for a fixed $\alpha > 0$
\begin{equation} \label{eq:La}
\quad L^{\alpha}=  L + \alpha I, 
\end{equation}
where $I\in \R^{dN \times dN}$ denotes the identify matrix. { To
  see why this shift is used, observe that due to the existence of the
  zero frequency mode, the measure $\nu_N(\bd{q})$ corresponding to
  the spring potential $ \frac{1}{2} \bd{q} \cdot L \bd{q}$ is not
  normalizable even in the finite dimensional cases. Adding the extra
  harmonic trapping potential makes the shifted spring potential
  $ \frac{1}{2} \bd{q} \cdot L^\alpha \bd{q}$ grow quadratically at
  infinity. We will see below that with this shift the corresponding
  measure as well as its infinite dimensional limit are properly
  defined.  }

Note that $L$ can be viewed as a second order finite difference
approximation to the operator $-\Delta$ with the periodic boundary
condition for an equidistant partition of the interval $[0,
\beta]$.
Thus as $N \rightarrow \infty$, $L^\alpha $ converges to the
differential operator
\begin{equation} \label{eq:cLa}
\mathfrak L^\alpha = \bigl( - \partial_{\tau \tau} + \alpha \bigr) \text{Id}
\end{equation}
with periodic boundary conditions and $\text{Id}$ is the identity operator from $\R^d$ to $\R^d$. 

In terms of $L^{\alpha}$, we define a new Gaussian measure 
\[
\ud \nu_N^{\alpha}(\bd{q})  \propto \exp\Bigl( - \frac{1}{2} \frac{\beta}{N} \bd{q}  \cdot L^{\alpha} \bd{q} \Bigr) \ud \bd{q},
\]
and we can rewrite the measure $\mu_N$ as 
\begin{equation}
  \ud \mu_N(\bd{q}) \propto \exp \Bigl( - \frac{\beta}{N}\sum_{i=1}^N U^{\alpha}(q_i) \Bigr) \ud \nu_N^{\alpha}(\bd{q})
\end{equation}
with $U^{\alpha}(q)= V(q)- \frac{\alpha}{2}|q|^2$, which we assume to
grow to infinite as $\abs{q} \to \infty$ (with $\alpha$ chosen
sufficiently small).  As $N \to \infty$, we have
\begin{equation*}
  \frac{\beta}{N} \bd{q} \cdot L^{\alpha} \bd{q} \to 
\int_0^{\beta} \abs{\partial_\tau \mathfrak{q}(\tau)}^2 + \alpha 
\abs{\mathfrak{q}(\tau)}^2 \ud \tau
\end{equation*}
and correspondingly the Gaussian measure $\nu_N^{\alpha}$ converges to
a well-defined infinite dimensional Gaussian measure $\nu^{\alpha}$
with mean $0$ and covariance operator given by
$\mathfrak C^\alpha=(\mathfrak L^\alpha)^{-1}$. The infinite
dimensional limit of $\mu_N$ is given therefore by
\begin{equation}
  \begin{aligned}
    \ud \mu & \propto \exp\Bigl( - \int_0^{\beta} U^{\alpha}(\mathfrak{q}(\tau)) \ud \tau \Bigr) \ud \nu^{\alpha} \\
    & \propto  \exp\Bigl( - \int_0^{\beta} \bigl( V(\mathfrak{q}(\tau)) -
    \frac{\alpha}{2} \abs{\mathfrak{q}(\tau)}^2 \bigr) \ud \tau \Bigr)
    \ud \nu^{\alpha}. 
  \end{aligned}
\end{equation}
In particular, the measure $\mu$ is absolutely continuous with respect to  the reference measure $\nu^\alpha$. 
Thus, we have 
\begin{align*}
  \langle \wh{A} \rangle &=\mathbb{E}_{\mu}\Bigl[  \frac{1}{\beta}\int_0^{\beta} A( \mathfrak q(\tau)) \ud \tau  \Bigr] \\
                         & = \frac{1}{\mathcal Z} \int  e^{ - \int_0^{\beta} U^{\alpha}(\mathfrak{q}(\tau)) \ud \tau}  
                           \frac{1}{\beta} \int_0^{\beta} A( \mathfrak q(\tau)) \ud \tau\, \nu^{\alpha}(\mathrm{d} \mathfrak q). 
\end{align*}
This gives a rigorous path integral representation of the thermal
average for quantum systems. To approximate $\langle \wh{A} \rangle$,
it thus suffices to construct some ergodic sampling schemes for the
distribution $\mu$ on the path space, which we discuss next.

\subsection{Overdamped Langevin sampling} \label{sec:overdamp}

We now consider the overdamped Langevin sampling for ring-polymers,
and the continuum limit.  We will then introduce the preconditioning
techniques to the continuum limit of the sampling equation. Those
techniques have been extensively used for efficient sampling of
diffusion bridges as in e.g., \cite{HSVW05,HSV07,BRSV08,BPSS11}, but
to our best knowledge, not for sampling of random loops arising from
the path integral representation.

\smallskip

Recall that the probability distribution $\mu_N$ can be sampled using the overdamped Langevin dynamics 
\begin{equation}
\begin{aligned}  \label{eq:olang}
\ud \bd q & =( -L \bd q - \nabla_{\bd q} V_N ) \ud t +
\sqrt{\frac{2}{\beta_N}} \ud \bd B \\
& =( -L^{\alpha} \bd q - \nabla_{\bd q} U_N^{\alpha} ) \ud t +
\sqrt{\frac{2}{\beta_N}} \ud \bd B,
\end{aligned}
\end{equation}
where we have used the short-hand notations $\beta_N = \beta / N$,
$V_N (\bd q)= \sum_{i=1}^N V(q_i)$,
$U_N^{\alpha} (\bd q)= \sum_{i=1}^N U^{\alpha}(q_i)$,
%$\nabla_{\bd q} V_N = ( \nabla V(q_1), \cdots, \nabla V(q_N) )^T$,
and $\bd B$ is an $\RR^{dN}$ dimensional Brownian motion with
independent components.  We observe that when $N \gg 1$, the forcing
term from the spring potential becomes quite stiff, which prevents the
use of large time steps in numerical integration of the sampling
trajectory. The restriction becomes more severe as the number of bead
increases, and thus some preconditioning treatments are desired to
enhance sampling efficiency by allowing the use of large time steps.

To gain insights for the design of the preconditioner, let us consider
the most difficult scenario as the limit $N \to \infty$, where the
Langevin dynamics converges to the following stochastic partial
differential equation (SPDE)
\begin{align} 
  \ud  \mathfrak q &= \partial_{\tau \tau } \mathfrak  q  \ud t - \nabla V( \mathfrak  q) \ud t  + \sqrt{2 } \ud w \nonumber \\
&= - \mathfrak L^\alpha \mathfrak  q \ud t - \nabla U^\alpha (\mathfrak  q) \ud t  + \sqrt{2 } \ud w, \label{eq:ospde}
\end{align}
with periodic boundary conditions in $\tau$ and $\ud w$ denotes the
space-time white noise. Here, with slight abuse of notations,
$\mathfrak q(t):\, \R^+ \to \mathcal{L} \R^{d}$ is the trajectory of
loops in the configuration space $\R^d$, where $t$ is viewed as the
temporal variable for the trajectory and $\tau$ is the parameter for
the loops. The SPDE samples loops, the continuum analog of ring
polymers as $N \to \infty$, with the invariant measure given by $\mu$,
the infinite dimensional limit of $\mu_N$.

The reason to consider the continuum limit is that the overdamped
Langevin equation \eqref{eq:olang} can be viewed as a finite
dimensional approximation of the SPDE
\eqref{eq:ospde}. Thus, the preconditioning techniques for the SPDE
will shed light on the choice of preconditioner for the overdamped
Langevin equations. In particular, as the stiffness comes from the
differential operator $\mathfrak L^{\alpha}$, a natural idea is to 
precondition the SPDE using the
inverse of $\mathfrak L^\alpha$, namely,
\begin{equation}
  {\ud  \mathfrak  q(t,\tau)}= -\mathfrak q  \ud t - (\mathfrak L^{\alpha})^{-1} \nabla U( \mathfrak q) \ud t 
  + \sqrt{2  (\mathfrak L^{\alpha})^{-1} } \ud w.
\end{equation}
Similar preconditioners have been used for diffusion bridge sampling,
see e.g., \cite{HSVW05,HSV07}.  For finite number of beads, the
natural analog of such strategy is to precondition the overdamped
Langevin equation \eqref{eq:olang} with the inverse of
$L^{\alpha}$. Hence the preconditioned equation is given by
\begin{equation}
\ud \bd q =( - \bd q - (L^{\alpha})^{-1} \nabla_{\bd q} U_N^{\alpha} ) \ud t +
\sqrt{\frac{2(L^{\alpha})^{-1}}{\beta_N}} \ud \bd B.
\end{equation}
It is easy to check that the preconditioned equation takes the same
invariant measure, while the stiff term $L^{\alpha} \bd{q}$ is now
replaced by a linear damping term, which is much easier to handle
numerically. For completeness, the inverse of $\mathfrak L^\alpha$ and
its finite dimensional approximation are discussed in Appendix~\ref{sec:cov}.

%\blue{
%We remark that although several preconditioning ideas have previously been proposed to deal with the spring potential based on certain changes of variables and adjusting the masses of the auxiliary momentum variables, see e.g. \cite{TBMK93,CM93,MarklandManolopoulos:08,Matsubara}, those methods, however, can not be applied to the overdamped case.
%} \jl{why do we want to talk about this? I guess our focus is the Langevin anyway?}

%\blue{We remark that, the matrix $L$ and its continuum limit are not invertible due to the zero eigenvalue and thus cannot directly be used to precondition the sampling equation. In the previous works, see e.g. \cite{TBMK93,CM93,Matsubara}, several preconditioning ideas were proposed to deal with the spring potential directly based on certain change of variables, but the potential inevitably vanishes in the first mode. We instead choose to precondition with the inverse of $L^{\alpha}$, and thus the damping effect persists for all modes. We shall revisit this issue in the detailed comparison with the normal modes, the staging coordinates and the Matsubara modes in Section \ref{sec:comp}.}\jl{why it is put here?}

For ring polymer representation, compared with the overdamped
sampling, sampling using the underdamped Langevin dynamics is more
efficient and hence much more popular, see e.g.,
\cite{CeriottiParrinelloMarklandManolopoulos:10,LiuLiLiu:16}.  The
preconditioning strategy in the overdamped case discussed can be
extended to the underdamped case.

\subsection{Underdamped Langevin sampling} \label{sec:underdamp}

Let us now consider the preconditioning of underdamped Langevin
sampling of ring-polymers, again by the view point of taking the continuum
limit.  For path-integral molecular dynamics, auxiliary momentum
variables with artificial masses are introduced to improve the sampling
efficiency. In the augmented state space of position and momentum of
ring polymer beads, the thermal average is given by
\begin{equation}\label{eq:ensembleavgA} 
  \langle \wh{A} \rangle =
  \int_{\mathbb R^{dN} \times \mathbb{R}^{dN}}
  \frac{1}{N}\sum_{i=1}^N A(q_i) \, \pi_N(\mathrm{d}\bd{q} \ud \bd{p})+ \Or(N^{-2}),  
\end{equation}
with the Gibbs distribution
\begin{equation}
\pi_N (\mathrm{d}\bd{q} \ud \bd{p}) = \frac{1}{\mathcal Z'_N} e^{- \beta_N H_N ( \bd{q}, \bd{p})} \ud \bd{q} \ud \bd{p}
\end{equation} 
and the Hamiltonian
\begin{equation}\label{eq:Ham1}
  \begin{aligned}
    H_N (\bd q, \bd p) & = \frac{1}{2} \bd p\cdot M^{-1}
    \bd p + \sum_{i=1}^N \left[ \frac{|q_i - q_{i+1}|^2}{2 \beta_N^2}
      + V(q_i) \right]\\
    & = \frac 1 2 \bd p \cdot M^{-1} \bd p + \frac 1
    2 \bd q \cdot L^{\alpha} \bd q + U_N^\alpha (\bd q),
\end{aligned}
\end{equation}
where $\bd{q} = (q_1, \cdots, q_N)$ and $\bd{p} = (p_1, \cdots, p_N)$
are the position and momentum of the beads with the convention
$q_{N+1}=q_1$, and $\mathcal Z'_N$ is the normalization constant of
the probability distribution $\pi_N$. In the Hamiltonian, $M$ is a
positive definite fictitious mass matrix for the auxiliary momentum
variables $\bd{p}$, which should not be confused with the physical
mass of the quantum particles (which has been chosen to be $1$ from
the beginning).

It is clear from the definition that the distribution $\pi_N$ takes a product form 
\begin{equation}\label{eq:piN}
  \pi_N (\mathrm{d}\bd{q} \ud \bd{p}) = \frac{1}{\mathcal{Z}_N'} 
  \bigl( e^{-\beta_N S_N(\bd{q})} \ud \bd q \bigr) \bigl( e^{-\frac{\beta_N}{2} \bd{p} \cdot M^{-1} \bd{p}} \ud \bd p\bigr). 
\end{equation}
In particular, the marginal distribution with respect to $\bd{q}$,
which agress with $\mu_N$, is independent of the choice of the
fictitious mass matrix $M$, as long as it is positive definite. Thus
many choices can be made for the benefit of sampling efficiency.  One
common choice in the literature of path-integral molecular dynamics is
to take $M$ a constant multiple of the identity matrix $M = m I$,
where $m$ is a scalar, see
e.g., \cite{HabershonManolopoulosMarklandMiller:13}.  With this choice,
we obtain the following Langevin equation associated with the
Hamiltonian \eqref{eq:Ham1},
\begin{equation}\tag{Lang}\label{eq:Lang}
\begin{aligned} 
  {\ud \bd q}&=  \frac{\bd p}{m} \ud t;\\
  {\ud  \bd p}&= -L \bd q \ud t -\nabla_{\bd q} V_N \ud t \gamma  \bd p\ud t + \sqrt{\frac{2 \gamma  m}{\beta_N}} \ud \bd B,
\end{aligned}
\end{equation}
where $\gamma > 0$ is the friction parameter and $\bd B$ denotes a
vector of $dN$ independent Brownian motion.  {Diagonal mass
  matrix with variable entries have also been explored to adjust the
  mode frequencies \cite{TBMK93,CM93}.}  We will compare the Langevin
equation \eqref{eq:Lang} with its variants introduced in the sequel.

When the number of beads $N$ is large, the forcing $-L \bd q$ becomes
stiff which prevents the use of large time steps in numerically
integrating \eqref{eq:Lang} (in fact, as $N \rightarrow \infty$, the
allowed time step size decreases to $0$). This can be seen as in the
Hamiltonian dynamics, $\bd q(t)$ consists of both $O(1)$ and high
frequency modes, with the latter induced by the stiff spring potential
between the beads. Preconditioning of \eqref{eq:Lang} is thus required
when $N \gg 1$ for efficient sampling.

Using $L^{\alpha}$ and $U^{\alpha}$, the momentum part of \eqref{eq:Lang} can be rewritten as 
\begin{equation}
  {\ud  \bd p}=- L^\alpha \bd q \ud t -\nabla_{\bd q} U_N^{\alpha} \ud t - \gamma  \bd p\ud t + \sqrt{\frac{2 \gamma  m}{\beta_N}} \ud \bd B.
\end{equation}
Therefore, similar to what has been done in the overdamped case, we may use $(L^{\alpha})^{-1}$
%\jl{the equation below should already be used in the overdamped case}
%\begin{equation} \label{eq:MC}
%(L^\alpha)^{-1}=(L +\alpha I)^{-1},
%\end{equation}
to precondition the system. It is straight-forward to verify that the following preconditioned Langevin equation samples the same invariant measure $\pi_N$ (for completeness, a derivation is given in Appendix~\ref{app:inv}): 
\begin{equation}\tag{pLang}\label{eq:pLang}
  \begin{aligned}
    {\ud \bd q} &=  {\frac 1 m (L^\alpha)^{-1} \bd p} \ud t;\\
    {\ud  \bd p} &= - \bd q \ud t - (L^\alpha)^{-1} \nabla_{\bd q} U_N^{\alpha} \ud t - \gamma (L^\alpha)^{-1} \bd p\ud t + \sqrt{\frac{2 \gamma m (L^\alpha)^{-1}}{\beta_N}} \ud \bd B.
  \end{aligned}
\end{equation}
However, taking the mass matrix as a scalar multiple of the identity
matrix has issues when $N \to \infty$. In that case the
distribution $\pi_N$ as in \eqref{eq:piN} does not have a limit, since the momentum part of the resulting measure is not normalizable. On
the other hand, the limit of $\mu_N$, the marginal distribution in
$\bd{q}$, does exist and is given by $\mu$. As a result, the limiting
process of \eqref{eq:pLang} is not well defined and we may encounter
trouble when using the dynamics \eqref{eq:pLang} to sample for large
$N$. 

This issue can be overcome by a proper choice of the mass matrix.
Recall that after all the choice is arbitrary for finite dimensional
systems as long as the mass matrix is positive definite.  In particular, inspired by the work \cite{BPSS11} which
considers hybrid Monte Carlo methods in infinite dimension, we make
the choice $M=L^\alpha$, which leads to the following mass-modified
Langevin (mmLang) dynamics 
\begin{equation}\tag{mmLang} \label{eq:mmLang}
\begin{aligned} 
  {\ud \bd q}&=  {(L^\alpha)^{-1} \bd p} \ud t;\\
    {\ud  \bd p}&= -L^\alpha  \bd q \ud t -\nabla_{\bd q}  U_N^{\alpha} \ud t  - \gamma  \bd p\ud t + \sqrt{\frac{2 \gamma  L^\alpha}{\beta_N}} \ud \bd B.
\end{aligned}
\end{equation}
The choice of the mass matrix $M = L^{\alpha}$ is made, since
  after a change of variable, it leads to a sampling measure with
  well-defined infinite dimensional limit; and also gives rise to a
  preconditioned Langevine dynamics with superior properties
 %(see also Section~\ref{sec:comp} where the choice of $M$ is further elaborated). 
 Let us multiply the momentum part of
\eqref{eq:mmLang} by $(L^\alpha)^{-1}$ and further introduce the
velocity variable $\bd v = M^{-1} \bd p=(L^\alpha)^{-1} \bd p$, we
obtain the preconditioned mass-modified Langevin (pmmLang) dynamics
\begin{equation}\tag{pmmLang} \label{eq:pmmLang}
\begin{aligned} 
  {\ud \bd q}&=  { \bd v} \ud t;\\
  {\ud  \bd v}&= -  \bd q \ud t - (L^\alpha)^{-1} \nabla_{\bd q}  U_N^{\alpha} \ud t  - \gamma  \bd v\ud t + \sqrt{\frac{2 \gamma  (L^\alpha)^{-1}}{\beta_N}} \ud \bd B.
\end{aligned}
\end{equation}
In the limit $N \to \infty$, the system \eqref{eq:pmmLang} converges
to the following SPDE
\begin{align} \label{eq:pmmlspd}
  {\ud  \mathfrak q(t,\tau)}&=  {  \mathfrak v } \ud t; \nonumber\\ 
  {\ud  \mathfrak v(t,\tau)}&= - \mathfrak  q \ud t - \mathfrak C^\alpha \nabla_{\mathfrak  q}  U^{\alpha} \ud t  - \gamma \mathfrak  v \ud t + \sqrt{{2 \gamma \mathfrak C^\alpha }} \ud w, 
\end{align}
where $\ud w$ is the space-time white noise, and thus
$\sqrt{{ \mathfrak C^\alpha }} \ud w$ is the cylindrical
$\mathfrak C^\alpha$-Wiener process in probability terms (see e.g.,
\cite{DZZ96,Hairer16}).  Under some technical assumptions on $V$, it
can be proved that the phase space distribution
\[
 \pi'_N (\bd q, \bd v)  \propto \exp\left( - \beta_N S_N  (\bd q) \right)  \exp\left( - \beta_N \frac  1 2 \bd v \cdot L^{\alpha} \bd v \right)
\]
converges to a well defined probability distribution $\pi'$ as
$N \rightarrow \infty$, with marginal distribution
$\mu{\mathrm{d} \mathfrak q}$.  Moreover, the the SPDE system
\eqref{eq:pmmlspd} takes
$\pi'(\mathrm{d} \mathfrak q \ud\mathfrak v)$ as the invariant measure. We
will not go into the details here.

To conclude this section, we remark that both \eqref{eq:pLang} and
\eqref{eq:pmmLang} have preconditioned the underdamped Langevin
dynamics: in either system, the stiff forcing term $L^{\alpha} \bd q$
is replaced by a linear damping force term. However, only
\eqref{eq:pmmLang} has a well-defined continuum limit as the number of
beads goes to infinity. We shall study their numerical performances in
Section \ref{sec:num} when large number of beads are needed.

%%%%%%%%%%%%%%%%%%%%%%%%%%%%%%%%%%

\section{Normal modes analysis and Convergence in the continuum limit} \label{sec:conv}

%\subsection{Normal modes analysis}

 %And we remark the comparison may be extended to other sampling methods, such as PIMD, HMC and LSC-IVR, which we reserve for future studies.
%%%%%%%%%%%%%%%%
%
%\section{Connections with normal modes {, staging coordinates} and Matsubara modes} \label{sec:comp}
%
% Before we numerically study the proposed sampling dynamics,

In this section we compare our preconditioning approaches with the more familiar
  normal modes, staging coordinates and Matsubara modes in the
  literature of path integral molecular dynamics.  Since those
  representations are introduced mostly for the Hamiltonian part (but
not the thermostating such as Langevin), we will compare them with the
proposed approach in the context of the Hamiltonian
dynamics.  Also, we show that when the potential is quadratic, 
the continuum limit of \eqref{eq:pmmLang} system converges exponentially fast to its invariant measure, which implies that the \eqref{eq:pmmLang} system converges to the invariant measure exponentially with a convergence rate that is independent of the system size. 

%And we remark the comparison may be extended to other sampling methods, such as PIMD, HMC and LSC-IVR, which we reserve for future studies.

\subsection{Normal modes, staging coordinates and preconditioning} \label{sec:com}

Recall that with the artificial mass matrix $M$, the Hamiltonian in
PIMD is given by
\begin{align}
H_N & =  \frac 1 2 \bd p \cdot  M^{-1} \bd p + \frac 1 2 \bd q \cdot L \bd q +V_N (\bd q) \cr 
& =:  H^\alpha+ U^\alpha_N,
\end{align}
where
$H^\alpha = \frac 1 2 \bd p \cdot M^{-1} \bd p + \frac 1 2 \bd q \cdot L^\alpha \bd
q$
is the harmonic part.  The corresponding Hamiltonian dynamics is given
by
\begin{align}
\frac{\ud}{\ud t} \bd q &= M^{-1} \bd p; %\label{eq:Hqo}
\nonumber \\
\frac{\ud}{\ud t} \bd p & =  - L \bd q - \nabla_{\bd q} V_N(\bd q) = - L^\alpha \bd q - \nabla_{\bd q} U^\alpha_N(\bd q).   \label{eq:Hp} 
\end{align}
Let us introduce the normal modes $(\wt p_k, \wt q_k)$ given by (we assume $N$ is odd to simplify the algebra)
\[
\widetilde p_k = \sum_{j=1}^N p_j D_{jk} \quad \text{and} \quad \widetilde q_k = \sum_{j=1}^N q_j D_{jk}, \quad
 k=0, \pm 1, \cdots, \pm \frac{N-1}{2}, 
\]
with the transformation matrix defined as
\begin{equation*}
  D_{jk}=  
  \begin{cases}
    \sqrt{1/N}, & k=0,  \\
    \sqrt{2/N} \sin\bigl( \frac{2 \pi jk}{N} \bigr), & 0<k\le \frac{N-1}{2}, \\
    \sqrt{2/N} \cos \bigl( \frac{2 \pi j k}{N}\bigr), & -
    \frac{N-1}{2} \le k< 0.
  \end{cases}
\end{equation*}
It is easy to check that $D$ is an orthogonal matrix, hence the
inverse transform is given by
\[
p_j = \sum_{k=-(N-1)/2}^{(N-1)/2}D_{jk} \widetilde p_k,   \quad  q_j =  \sum_{k=-(N-1)/2}^{(N-1)/2}D_{jk} \widetilde q_k. 
\]
Using matrix notations, we have
\[
\widetilde {\bd q} = D^T \bd q, \quad  \widetilde {\bd p} = D^T \bd p, \quad \text{and} \quad \bd q = D \, \widetilde {\bd q}, \quad \bd p = D \, \widetilde {\bd p}. 
\]

In the normal mode representation, when the mass matrix is chosen as
$M=m I$, $H_N$ can be written as 
\[
H_N(\widetilde {\bd q}, \widetilde {\bd p})= \sum_{k=-(N-1)/2}^{(N-1)/2} \left[ \frac{\widetilde p_k^2}{2m} +  \frac{1}{2} \omega_k^2 \widetilde q_k^2 \right]+\widetilde{V}_N(\widetilde {\bd q}),
\]
where 
\begin{align*}
\omega_k & =\frac{2}{\beta_N} \sin\left( \frac{k\pi}{N}\right),\\ 
\widetilde{V}_N(\widetilde {\bd q}) & = \sum_{l=1}^N V\left(  \sum_{k=-(N-1)/2}^{(N-1)/2}D_{jk}  \widetilde q_k \right).
\end{align*}
{In particular, we observe that $\omega_0=0$, which is consistent with the fact that the $L$ matrix has the lowest eigenvalue $0$. The corresponding spatial variable
\[
\widetilde q_0 =\frac{1}{N} \sum_{j=1}^N q_j
\]
is the centroid of the ring polymer. In Hamiltonian dynamics, the spring potential vanishes for the centroid mode, and the associated momentum  changes according to the averaged force due to the external potential 
\[
\frac{\ud}{\ud t} \widetilde q_0= \frac 1 m \widetilde p_0,\quad \frac{\ud}{\ud t} \widetilde p_0 = - \frac{1}{N} \sum_{j=1}^N V'(q_j(\widetilde {\bd  q})).
\]
The force from the spring potential shows up in momentum equation for the rest of the normal modes, which becomes dominant when $|k|$ is large. 
} 
% \jl{perhaps giving the equation for the centroid after the whole
%   system as a remark? why we want to emphasize that here?} \zz{I wanted to clarify the role of the two potentials in the normal mode representation, but I am not sure if it is a good idea...} 

The Hamiltonian dynamics of normal modes reads
\begin{align}
\frac{\ud}{\ud t} \widetilde{\bd q} &= \frac{1}{m} \widetilde{\bd p}; \nonumber \\
\frac{\ud}{\ud t} \widetilde{\bd p} & =  - D^T L D \widetilde{ \bd q} -  D^T \nabla_{\bd q} V(D \widetilde{ \bd q}) \\
& =  - D^T L^\alpha D \widetilde{ \bd q}-  D^T \nabla_{\bd q} U^\alpha_N(D \widetilde{ \bd q}). \nonumber
\end{align}
The choice of $D$ ensures that $D^TLD$ is a diagonal matrix with
entries $ \omega_k^2$, and also $D^T L^\alpha D$ is a diagonal matrix
with entries $\omega_k^2+ \alpha$. Therefore, with the use of the
normal modes, the stiff part of the Hamiltonian dynamics, namely
$L \bd q$, is diagonalized. This of course facilitates the design of
accurate numerical scheme, but it does not, however, reduce the
stiffness of the dynamics, because $D^T L D$ and $L$ share the same
eigenvalues. When the dimension is large, the fast modes lead to
severe stability constraints which prevent efficient sampling.

With the normal mode transformation, we obtain that the harmonic part
of the original Hamiltonian (i.e., when we neglect $U^\alpha$) contains modes with frequency from $\sqrt{\alpha/m}$ to
$\mathcal O(N)$. { We remark that, the diagonal mass matrix  with variable components have been considered in previous works to tune the frequencies of the normal modes (see e.g. \cite{CM93}), but centroid mode is treated differently from the rest of the normal modes since the spring potential vanishes for that mode. In particular, when $k\ne 0$, the masses of the normal modes can be chosen such that effectively those modes have the same frequency.} %\jl{I don't understand this sentence} \zz{any better?}

%\subsection{The staging coordinates representation }

{
The staging coordinates $( u_k )$ for the position variables $(q_k)$ introduced in \cite{TBMK93} lead to another useful transformation in the ring polymer representation, which is given by 
\begin{equation}\label{eq:q2u}
u_1=q_1, \quad u_k=q_k- \frac{(k-1)q_{k+1}+q_1}{k}, \quad
k=2,\cdots, N.
\end{equation}
The inverse transform is 
\begin{equation}\label{eq:u2q}
q_1=u_1, \quad q_k=u_1+ \sum_{k'=k}^N \frac{k-1}{k'-1}u_{k'}, \quad k=2,\cdots, N.
\end{equation}
 With matrix notations, we have 
\[
\bd u= D_1 \bd q, \quad \bd q = D_2 \bd u = (D_1)^{-1} \bd u,
\]
where the transform matrices are defined according to \eqref{eq:q2u} and \eqref{eq:u2q} respectively.
With the staging coordinates, the Hamiltonian can be rewritten as
\begin{align}\label{eq:Ham_staging}
H_N & = \sum_{k=1}^N  \left[ \frac{ \tilde v_k^2}{2\tilde{m}_k} +V(q_k(\mathbf u))\right]+\sum_{k=2}^N \frac{m_k}{2 \beta_N^2} u_k^2 \\
& =\frac 1 2  \tilde {\bd v} \cdot \widetilde{\bd M}^{-1} \tilde {\bd v} + \frac 1 {2 \beta_N^2} \bd u \cdot \bar {\bd M} \bd u+\sum_{k=1}^N V(q_k(\mathbf u)). \nonumber 
\end{align}
Here, $m_k=\frac {k}{k-1}$ for $k=2,\cdots,N$, $\tilde v_k$ are the auxiliary momentum variables for the staging coordinates, and $\tilde{m}_k>0$ are artificial masses to be prescribed.  The $\widetilde{\bd M}$ matrix is a diagonal matrix with $\tilde{m}_k$ as components, and $\bar {\bd M}$ is also diagonal with $0$ as the first component and $m_k$ for the rest. Note that the staging transform diagonalize the spring potential in the following sense
\[
 \frac{1}{\beta_N^2} \bar {\bd M} = D_1 L D_2=D_1 L (D_1)^{-1}.
\]
The associated Hamiltonian dynamics of the harmonic part are given by (i.e. dropping $V$)
\begin{align}
\frac{\ud}{\ud t} {\bd u} &= \widetilde{\bd M}^{-1}  \tilde{\bd v}; \label{eq:u}\\
\frac{\ud}{\ud t} \tilde{\bd v} & =  - \frac{1}{\beta_N^2} \bar {\bd M} \bd u. \label{eq:tv}
\end{align}
We observe that, the spring potential vanishes for the first staging variable $u_1$, and the harmonic part of the Hamiltonian all have $O(N)$ frequency except for the first mode.
%Note that, when $\beta_N \ll 1$, the equations of motion exhibits two scales in dynamics, i.e. the $\tilde{\bd v}$ variables are changing much faster than the $\bd u$ variables.
%\jl{is this actually true? Consider for example
%\begin{align*}
%  & \frac{\ud}{\ud t} u = v \\
%  & \frac{\ud}{\ud t} v = - \beta_N^{-2} u 
%\end{align*}
%the matrix is 
%\begin{equation*}
%  \begin{pmatrix} 
%    0 & 1 \\
%    - \beta_N^{-2} & 0 
%  \end{pmatrix}
%\end{equation*}
%I think the eigenvalues in fact have the same magnitude $\beta_N^{-1}$? It seems to suggest that we just need to adjust the time step size properly and the behavior of this dynamics would be similar to the one of \eqref{eq:pmmLang}?
%}
Similar to the normal mode representation, authors in \cite{TBMK93,CM93,LiuLiLiu:16} proposed to use the artificial mass matrix  $\widetilde{\bd M}$ to adjust all the modes except the first one to the same frequency. This is be achieved simply by choosing $\tilde{m}_k=m_k$ for $k=2,\cdots,n$.
}

{
Note that, the auxiliary momentum variables $ \tilde {\bd v}$ for the staging variables $\bd u$ are not the same as the momentum variables $\bd p$ for the (Cartesian) position variables $\bd q$. To find the connections, we first multiply \eqref{eq:tv} from the left by $D_2$, and we get
\[
\frac{\ud}{\ud t} D_2 \tilde{\bd v}  =  - L \bd q. 
\]
By comparing with the harmonic part of \eqref{eq:Hp},
this equations implies the connection between $\tilde{\bd v}$ and $\bd p$,
\[
D_2 \tilde{\bd v} =\bd p.
\]
And with this change of variable, if we multiply \eqref{eq:u} from the left by $D_2$, we get
\[
\frac{\ud}{\ud t} {\bd q} =( D_2 \widetilde{\bd M} D_1)^{-1}  {\bd p}.
\]
This shows, the staging coordinates actually implies choosing a non-diagonal mass matrix $M_{\text{eff}}$, which is given by
\[
M_{\text{eff}}=D_2 \widetilde{\bd M} D_1 = \beta_N^2 L + \tilde{m}_1  E_1. 
\]
%\zz{I removed a * in the line above.}
where $E_1$ is a square matrix with the elements in the first column equal $1$, and the rest of the elements are $0$.
}

{
From the analysis above, we see that the system with the staging coordinate representation and the preconditioned mass-modified system introduced in Section~\ref{sec:underdamp} share a few common features: both can be viewed as choosing non-diagonal mass matrices which are accompanied by changes of variables, such that the harmonic part of the dynamics are diagonalized. But the choices of the mass matrices and the changes of variables are different, and unlike the preconditional mass-modified system, the staging representation leads to {high frequency modes} when $\beta_N \ll 1$. We shall further compare those two approaches in terms of adjusting frequencies later in this section.
}

Let us now compare with the two proposed
preconditioned dynamics \eqref{eq:pLang} and \eqref{eq:pmmLang}.

The Hamiltonian part of the system \eqref{eq:pLang} is given by
\begin{align} 
  \frac{\ud }{\ud t}{ \bd q}&=  {\frac 1 m (L^\alpha)^{-1} \bd p};\\
    \frac{\ud }{\ud t} {\bd p}&=- \bd q - (L^\alpha)^{-1} \nabla_{\bd q} U_N^{\alpha} .
\end{align}
In the harmonic case that $U^{\alpha}=0$, it reduces to
\[
  \frac{\ud^2}{\ud t^2}{ \bd q}=-\frac 1 m (L^\alpha)^{-1} \bd q.
\]
Note that $L^{\alpha}$ has the smallest eigenvalue $\alpha$, therefore
$(L^\alpha)^{-1}/m$ has the largest eigenvalue $(\alpha m)^{-1}$ and
the smallest eigenvalue $\bigl ((\alpha+\omega_N^2 )m\bigr)^{-1}$.
Thus, with the help of the $\alpha$-regularization, the frequency of
the fastest mode has a uniform upper bound no matter how large $N$
is. In other words, increasing the number of beads only adds more slow
modes of the Hamiltonian system, which does not cause numerical
stability constraints. Although the frequency ratio between the
fastest mode and the slowest mode (i.e., condition number for the
linear system) is not changed in the system \eqref{eq:pLang}: while it
allows for larger step sizes, it also takes longer time for the slow
modes to equilibrate. In terms of sampling, the system
\eqref{eq:pLang} is still superior to the system  \eqref{eq:Lang}, since the
low-frequency modes in \eqref{eq:Lang} are more important for accurate
approximation of the thermal average, which are mapped to the
high-frequency modes in the system \eqref{eq:pLang} with a uniform
frequency upper bound $(\alpha m)^{-1/2}$.  In other words, the
preconditioner in system \eqref{eq:pLang} prioritize the modes which
matter the most in sampling, at the prize of slowing the less
important modes.

{
We next look at two mass modified Langevin systems. The Hamiltonian part of \eqref{eq:mmLang} reads
\begin{equation}
\begin{aligned} 
  {\ud \bd q}&=  {(L^\alpha)^{-1} \bd p} \ud t;\\
    {\ud  \bd p}&= -L^\alpha  \bd q \ud t -\nabla_{\bd q}  U_N^{\alpha} \ud t,
\end{aligned}
\end{equation}
which reduces to a collection of independent oscillators with the same
frequency in the harmonic case ($U^{\alpha} = 0$)
\begin{equation}\label{eq:osc}
  \frac{\ud^2}{\ud t^2}{ \bd q}=- \bd q.
\end{equation}
The Hamiltonian part of \eqref{eq:pmmLang}
\begin{align}
   \frac{\ud }{\ud t}{ \bd q}&=  { \bd v} \ud t;\\
    \frac{\ud }{\ud t}{\bd v}&= -  \bd q \ud t - (L^\alpha)^{-1} \nabla_{\bd q} U_N^{\alpha},
\end{align}
lead to the same oscillator equations \eqref{eq:osc} in the harmonic case ($U^{\alpha} = 0$).
Therefore, the harmonic part of the two Hamiltonians only contain
unitary oscillations, regardless of the number of the beads $N$. The
choice of the mass matrix $M=L^{\alpha}$ automatically adjusts all
modes to the same frequency, so that the stability constraint is
removed and at the same time greatly reduce the condition number. This
is the main advantage of the choice of the mass matrix as
$L^{\alpha}$.
}

{
Compared with the normal mode representation and the staging coordinate representation, the mass matrix choice $M=L^{\alpha}$ (which is not a diagonal matrix) can help to adjust all mode to the same frequency $1$ without any exception. We remark that, tuning all modes to the same frequency, however, is not sufficient for improved numerical stability. In fact, the harmonic part of the \eqref{eq:mmLang} system also have uniform frequency, but as we shall see in Section \ref{sec:stability}, this system actually shows worse performances in terms of stability. The change of variable from $\bd p$ to $\bd v$ is also necessary in the preconditioning process. The essential difference between the \eqref{eq:mmLang} system and the the \eqref{eq:pmmLang} system is that the invariant measure of $(\bd q, \bd v)$ for \eqref{eq:pmmLang} does have a well-defined infinite dimensional limit, while the invariant measure of $(\bd q, \bd p)$ for \eqref{eq:mmLang} does not.
}

From another perspective, we can also understand the choice of the
mass matrix $M=L^\alpha$ using the normal mode representation. Let us
multiply the Hamiltonian system \eqref{eq:Hq}-\eqref{eq:Hp} by $D^T$
from the left and get
\begin{align}
\frac{\ud}{\ud t} \widetilde{\bd q} &= D^T M^{-1} D \widetilde{\bd p}; \\
\frac{\ud}{\ud t} \widetilde{\bd p} & =  - D^T L^\alpha D \widetilde{ \bd q} -  D^T \nabla_{\bd q} U^{\alpha}(D \widetilde{ \bd q}).
\end{align}
Recall that $D^TL^{\alpha}D$ is a diagonal matrix with positive
entries $\lambda^{\alpha}_k= \omega_k^2+\alpha$.  To adjust all the
normal modes to the same frequency, we require $D^T M^{-1} D$
is also diagonal, and its diagonal entries are exactly
$(\lambda^{\alpha}_k)^{-1}$, namely, if $U^{\alpha}=0$,
\begin{equation} \label{eq:masschoice}
\frac{\ud}{\ud t} \widetilde q_k = (\lambda^{\alpha}_k)^{-1} \widetilde p_k, \quad \frac{\ud}{\ud t} \widetilde  p_k=- \lambda^{\alpha}_k \widetilde q_k.
\end{equation}
This clearly leads to the condition $
D^T M^{-1} D=(D^T L^\alpha D)^{-1}$,
which results in  the choice of the mass matrix $M=L^\alpha$.

Then, to precondition the system \eqref{eq:masschoice}, we introduce $\widetilde v=(\lambda^{\alpha}_k)^{-1} \widetilde p_k$, such that the Hamiltonian dynamics \eqref{eq:masschoice} reduces to
\[
\frac{\ud}{\ud t} \widetilde q_k =  \widetilde v_k, \quad \frac{\ud}{\ud t} \widetilde  v_k=-  \widetilde q_k.
\]
In the matrix form, the change of variable becomes
\[
\widetilde {\bd v} = D^T (L^\alpha)^{-1} D \widetilde {\bd q}.
\]
This gives 
\[
\bd v= D\widetilde {\bd v} =D D^T (L^\alpha)^{-1} D \widetilde {\bd q} = (L^\alpha)^{-1}\bd q,
\]
which recovers the choice of the velocity variable in
\eqref{eq:pmmLang}.

\subsection{Matsubara modes and the continuum limit} \label{sec:mats}

Let us now consider the Matsubara modes, introduced in
\cite{Matsubara,HeleWillattMuoloAlthorpe:15a,HeleWillattMuoloAlthorpe:15b}.
We denote the Matsubara coordinates
$\overline{ \bd q }=(\overline {q}_k)$,
$\overline{\bd p}=(\overline {p}_k)$,
$k= -\frac{K-1}{2},\cdots,\frac{K-1}{2}$ (we take $K$ as an odd number
for simplicity).  The Matsubara modes are defined to be the $K$ lowest
normal models in the ring polymer representation in the limit that the
number of beads (and hence normal modes) goes to infinity. The
Matsubara modes are thus expected to be connected to the continuum
limit.

As is shown in \cite{HeleWillattMuoloAlthorpe:15a}, the effective
Hamiltonian in the Matsubara coordinates is given by
\[
\overline{H}_K= \sum_{k=-(K-1)/2}^{(K-1)/2}  \left[ \frac{\overline {p}_k^2}{2m} +  \frac{1}{2} \overline{ \omega}_k^2 \overline{ q}_k^2 \right]+ \overline{V}_K (\overline {\bd q}),
\]
where 
\begin{align*}
  \overline{ \omega}_k &= \frac{2 k \pi}{\beta}= \lim_{N\rightarrow \infty} \omega_k, \\
  \overline{V}_K(\overline {\bd q}) &= \lim_{N \rightarrow \infty}\frac{1}{N}\sum_{\ell=1}^N V \Bigl(  \sqrt{N} \sum_{k=-(K-1)/2}^{(K-1)/2}  D_{\ell k}  \overline q_k \Bigr).
\end{align*}
Thus, the corresponding Hamiltonian dynamics reads
\begin{align}
\frac{\ud}{\ud t} \overline{\bd q} &= \frac{1}{m} \overline{\bd p};  \nonumber\\
\frac{\ud}{\ud t} \overline{\bd p} & =  - \Lambda \overline{ \bd q} -  \nabla_{\overline{\bd q}} \overline{V}_K, \label{eq:mats}
\end{align}
where $\Lambda$ is a diagonal matrix with entries $ \overline \omega_k^2$.

The Hamiltonian dynamics in the Matsubara modes can be viewed as a
finite dimensional approximation of the continuum limit. For
simplicity, we take the potential function $V=0$ and the mass matrix
$M=mI$, and similar to the previous analysis, the continuum limit of
the Hamiltonian dynamic is given by
\begin{align}
\frac{\partial}{\partial t} { \mathfrak q(t,\tau)} &= \frac{1}{m}\mathfrak  p (t,\tau); \nonumber\\
\frac{\partial}{\partial t}\mathfrak  p(t,\tau) & =  - \mathfrak L \mathfrak q (t,\tau)= \frac{\partial^2}{\partial \tau^2} \mathfrak q(t,\tau). 
\end{align}
The operator $\mathfrak L$ with periodic boundary condition has eigenvalues 
\[
\overline \lambda_k=(\overline \omega_k)^2=  \left( \frac{2 k \pi}{\beta} \right)^2, \quad k=0,1,\cdots
\]
where $\overline \lambda_0$ has one linearly independent
eigenfunction, denoted by $\phi_0(\tau)$, and for $k \ge 1$, each
$\overline \lambda_k$ is associated with two linearly independent
eigenfunctions, denoted by $\phi_{\pm k}(\tau)$.  Thus, if we take the
following finite dimensional approximation
\[
\mathfrak q(t,\tau) \approx \sum_{k=-K'}^{K'} \hat q_k (t) \phi_k(\tau), \quad \mathfrak p(t,\tau) \approx \sum_{k=-K'}^{K'} \hat p_k (t) \phi_k(\tau).
\]
By orthogonality of the eigenfunctions, we obtain the ODE system for the coefficients
\[
\frac{\ud}{\ud t} \hat q_k = \frac 1 m \hat p_k, \quad \frac{\ud}{\ud t} \hat p_k=- \overline \lambda_{|k|} \hat q_k.
\]
This system agrees with the Hamiltonian system in the Matsubara coordinates as in \eqref{eq:mats} when $V=0$ and $K'=\frac{K-1}{2}$. 

While a finite dimensional approximation as using the Matsubara modes
reduces the stiffness of the system.  The frequency of the Matsubara
mode still ranges from $0$ to $\overline \omega_{(K-1)/2}$. Thus, when
$K$ is chosen large to improve accuracy, it also brings in similar
stability constraints.

To better compare the mode frequencies in different models, we plot in
Figure~\ref{fig:frequency} the distributions of the frequencies of the
normal modes and the Matsubara modes, compared with the frequencies of
the modes in (the harmonic part of) the \eqref{eq:pLang} and
\eqref{eq:pmmLang} systems. % The frequency plots also offer another
% perspective of interpreting the two preconditioning methods.
For the diagonal-mass preconditioned system \eqref{eq:pLang}, the use
of $(L^\alpha)^{-1}$ as a preconditioner maps the the fast modes to
the slow modes, while the upper bound of the mapped frequencies is
independent of the number of beads.  For the mass-modified system
\eqref{eq:pmmLang}, all modes are mapped to unitary frequency. Hence,
with those two preconditioning approaches, we no longer suffer
stability constraints as the number of beads increases. Moreover,
since the condition number is reduced to $1$ for the
\eqref{eq:pmmLang}, we expect the preconditioned mass-modified system
works the best in sampling efficiency.

\begin{figure}[ht]
\begin{centering}
\includegraphics[scale=0.55]{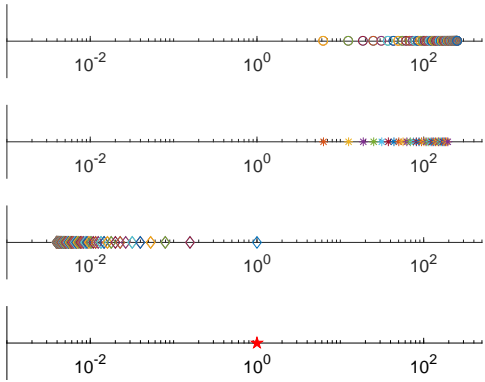} \\
\caption{Distribution of the mode frequencies in log scale. From top to bottom: normal models ($m=1$,  $N=63$); Matsubara modes ($m=1$, $K=31$); modes of the Hamiltonian part of the \eqref{eq:pLang} system ($\alpha=1$, $m=1$, $N=63$); modes of \eqref{eq:pmmLang} sytem ($\alpha=1$).}
\label{fig:frequency}
\end{centering}
\end{figure}

%\jl{I guess we need to compare our approach with the staging variable
%  approach better and discuss why ours has advantages? This question
%  seems unaddressed currently }
%\zz{I am not sure how to do a better comparison yet...}

\subsection{Convergence of the \eqref{eq:pmmLang} dynamics: the harmonic case}\label{sec:conv}

We have compared in the last section  our preconditioning approaches  with other existing preconditioning methods for the Hamiltonian dynamics. In the following, we switch gear back to the Langevin dynamics, and show that when the potential function is quadratic, the SPDE \eqref{eq:pmmlspd} converges to its invariant measure exponentially fast. We can also conclude from such analysis that the \eqref{eq:pmmLang} system exponentially converges  to its invariant measure with a rate that is independent of the bead number. 

In the harmonic case ($U^\alpha=0$), the   SPDE \eqref{eq:pmmlspd} reduces to
\begin{align} \label{eq:pmmlspdeh}
  {\ud  \mathfrak q(t,\tau)}&=  {  \mathfrak v } \ud t; \nonumber\\
  {\ud  \mathfrak v(t,\tau)}&= - \mathfrak  q \ud t  - \gamma \mathfrak  v \ud t + \sqrt{{2 \gamma \mathfrak C^\alpha }} \ud w. 
\end{align}
Formally, the invariant measure of the dynamics above is given by
\begin{equation} \label{eq:invh}
\Pi (\ud \mathfrak q, \ud \mathfrak p) =N(0, \mathfrak C^\alpha ) \times N(0, \mathfrak C^\alpha ).
\end{equation}
In the following, we aim to rigorously show that,  the law $\mathcal L (\mathfrak q(t), \mathfrak v(t))$ of the SPDE \eqref{eq:pmmlspdeh} converges to the invariant measure $\Pi$  exponentially fast with an explicit rate.

\begin{theorem} \label{thm:est}
Let $\mathcal L (\mathfrak q(t), \mathfrak v(t))$  be the law $\mathcal L (\mathfrak q(t), \mathfrak v(t))$ of the SPDE \eqref{eq:pmmlspdeh} with zero initial condition. Then 
\begin{equation} \label{est:conv}
\mathcal{W}_2(\mathcal{L}(\mathfrak q(t), \mathfrak v(t)), \Pi)  \lesssim e^{-{\theta (t)}},
\end{equation}
where the exponent $\theta(t)$ is given by
\[
\theta (t)=\left\{
\begin{split}
& \frac{\gamma }{2}t, & 0<\gamma<2,  \\
& t- \ln t, & \gamma=2,\\
&\frac{\gamma-\sqrt{\gamma^2-4}}{2}t, & \gamma>2.
\end{split}  \right.
\]
In \eqref{est:conv}, we use $a \lesssim b $ to denote that $a \leq C b $ for some generic constant $C$ which is independent of $t$ and $\gamma$.
\end{theorem}
\begin{proof}
We denote the eigen-pair of the covariance operator $\mathfrak C^\alpha$ by $\{\lambda_j,\phi_j \}_{j\in \mathbb N}$, i.e,
\[
\mathfrak C^\alpha \phi_j = \lambda_j^2 \phi_j, \quad j \in \mathbb N.
\]
Recall that $\mathfrak C^\alpha$ is the inverse of $\mathfrak L^\alpha$ and they share eigenfunctions, and  the eigenfunctions $\{ \phi_j \}$ form a basis of a separable Hilbert space, which we denote by  $\mathcal H$.  Then clearly $\mathfrak C^\alpha$ is a trace-class positive definite operator on $\mathcal H$, and $\forall x \in \mathcal H$, we have $x=\sum_{j=1}^{\infty} x_j \phi_j$. 

Due to the completeness of the  eigenfunctions $\{ \phi_j \}$, we write 
\[
\mathfrak q(t,\tau) = \sum_{j=1}^{\infty} q_{j}(t) \phi_j (\tau), \quad \mathfrak v(t,\tau) = \sum_{j=1}^{\infty} v_{j}(t) \phi_j (\tau).
\]
Then we rewrite the SPDE \eqref{eq:pmmlspdeh} in terms of the evolution of the Fourier coefficients as
follows
\begin{equation} \label{eq:coefj}
\begin{aligned}
& {\ud q_j} = v_j \ud t,\\
& {\ud v_j}  = - q_j \ud t -\gamma v_j \ud t  + \sqrt{2\gamma} \lambda_j {\ud W_j}.
\end{aligned}
\end{equation}
Here $\{W_j(t)\}$ is a sequence of i.i.d one-dimensional Brownian motions. Note that the invariant measure of above is
$$
\pi_j = \pi_{q,j}\otimes \pi_{v,j} := N(0,\sigma_{q,j}^2)\otimes N(0, \sigma_{v,j}^2)
$$
with  $\sigma_{q,j}^2 = \sigma_{v,j}^2 = \lambda_j^2$.

Observe that the system for each Fourier coefficient pair \eqref{eq:coefj}  is linear, and thus we can solve the system analytically. Let us define the matrix
$$
\mathbf A = \begin{pmatrix}
0 & 1 \\
- 1& -\gamma
\end{pmatrix}.
$$

We first consider the case when $\gamma^2 <4$.  And in this case, if we denote $\Delta = \gamma^2 -4$, then clearly $\Delta<0$, and the matrix $\mathbf A$ has two complex eigenvalues
\begin{equation}
\mu_\pm = -\frac{\gamma}{2} \pm i \frac{\sqrt{-\Delta}}{2}.
\end{equation}
Moreover, $\mathbf A$ can be diagonized as
$$
\mathbf A  = \mathbf P \Lambda \mathbf P^{-1} = \begin{pmatrix}
1 & 1 \\
\mu_+ & \mu_-
\end{pmatrix}
\begin{pmatrix}
\mu_+ &  \\
 & \mu_-
\end{pmatrix}
\begin{pmatrix}
1 & 1 \\
\mu_+ & \mu_-
\end{pmatrix}^{-1}.
$$
Setting $(q_j^\prime, v_j^\prime)^T = \mathbf P^{-1} (q_j, v_j)^T$, we have
\begin{equation} \label{eq:coefjp}
\begin{pmatrix} \ud {q_j^\prime}\\
\ud {v_j^\prime}
\end{pmatrix} = \begin{pmatrix}
\mu_+ &  \\
 & \mu_-
\end{pmatrix}
\begin{pmatrix}
q_j^\prime  \\
 v_j^\prime
\end{pmatrix} \ud t
+ \sqrt{2\gamma}\lambda_j \begin{pmatrix}
1 & 1 \\
\mu_+ & \mu_-
\end{pmatrix}^{-1}
\begin{pmatrix}
0  \\
 dW_j(t)
\end{pmatrix}.
\end{equation}
Note that, according to the assumption, the initial conditions for the coefficients $(q_j, v_j)$ are zeros, and thus the initial conditions for $(q_j^\prime, v_j^\prime)$ are also zeros (otherwise the effect of the initial condition decays exponentially with rate $\gamma/2$ anyway). And then the analytical solutions to \eqref{eq:coefjp} are given by
\begin{align*}
& q_j^\prime(t) = \frac{\sqrt{2\gamma} \lambda_j}{i \sqrt{-\Delta}}\int_0^t e^{\mu_+(t-s)}dW_j(s), \\
& v_j^\prime(t)  = - \frac{\sqrt{2\gamma} \lambda_j}{i \sqrt{-\Delta}}  \int_0^t e^{\mu_-(t-s)}dW_j(s).
\end{align*}
Then, with $(q_j, v_j)^T = \mathbf P (q_j^\prime, v_j^\prime)^T$, we obtain the analytical solution to the equation of Fourier coefficients \eqref{eq:coefj} with zero initial conditions,
\begin{equation} \label{sol:coefj}
\begin{aligned}
q_j (t) = & \frac{2 \sqrt{2\gamma}\lambda_j}{\sqrt{-\Delta}}\int_0^t e^{-\frac \gamma 2 (t-s)}\sin \left( \frac{\sqrt{-\Delta}(t-s)}{2} \right) \ud W_j (s), \\
v_j (t) = &- \frac{\gamma \sqrt{2\gamma}\lambda_j}{\sqrt{-\Delta}} \int_0^t e^{-\frac \gamma 2 (t-s)}\sin \left( \frac{\sqrt{-\Delta}(t-s)}{2} \right) \ud W_j (s) \\
 & + \sqrt{2 \gamma} \lambda_j \int_0^t e^{-\frac \gamma 2 (t-s)}\cos \left( \frac{\sqrt{-\Delta}(t-s)}{2} \right) \ud W_j (s).
\end{aligned}
\end{equation}
Clearly, $q_j(t)$ and $v_j(t)$ are mean-zero Gaussian random variable, and we denote their variances by $\Var_{q,j} (t)$ and $\Var_{v,j} (t)$, respectively.

By direct calculation, we get
\begin{equation}
\Var_{q,j} (t) =\frac{8 \gamma \lambda_j^2}{- \Delta} \int_0^t e^{- \gamma(t-s)} \sin^2  \left( \frac{\sqrt{-\Delta}(t-s)}{2} \right)  \ud s = \lambda_j^2 + \lambda_j^2 e^{-\gamma t} c_{j,1}.
\end{equation}
where
\[
c_{j,1}=\frac{4 \gamma}{-\Delta} \left( \frac{1}{- \gamma} - \frac{1}{4} \left( - \gamma \cos(\sqrt{-\Delta} t)+ \sqrt{-\Delta} \sin (\sqrt{-\Delta}t) \right) \right),
\]
and
\begin{equation}
\begin{aligned}
\Var_{v,j} (t) &=\int_0^t e^{- \gamma  (t-s)} \left( - \frac{\gamma \sqrt{2\gamma}\lambda_j}{\sqrt{-\Delta}} \sin \left( \frac{\sqrt{-\Delta}(t-s)}{2} \right) +\sqrt{2 \gamma} \lambda_j  \cos \left( \frac{\sqrt{-\Delta}(t-s)}{2} \right) \right)^2  \ud s \\
&= \lambda_j^2 + \lambda_j^2 e^{-\gamma t} c_{j,2},
\end{aligned}
\end{equation}
where 
\begin{align*}
c_{j,2} =& \frac{4}{\Delta}+ \frac{\gamma (4-2 \gamma^2)}{-4 \Delta} \left( \gamma \cos(\sqrt{-\Delta}t) + \sqrt{-\Delta} \sin (\sqrt{-\Delta}t) \right) \\& \quad + \frac {\gamma^2}{2\sqrt{-\Delta}} \left( \gamma \sin(\sqrt{-\Delta}t) + \sqrt{-\Delta} \cos (\sqrt{-\Delta}t) \right).
\end{align*}

Notice that,  when $t \rightarrow +\infty$, $\Var_{q,j} \rightarrow \lambda_j^2$ and $\Var_{v,j} \rightarrow \lambda_j^2$, and for fixed $\gamma$, there exists a constant $c_\gamma$, such that $\forall t >0$,
\[
|c_{j,1}| \le c_\gamma, \quad \text{and} \quad |c_{j,2}| \le c_\gamma.
\]
This 

Therefore, we have the following 
\begin{equation*}
\mathcal{W}_2^2\bigl(N(0, \Var_{q,j}(t)), \pi_{q,j}\bigr)  = \left(\sqrt{\Var_{q,j}(t)}  - \sigma_{q,j} \right)^2 \leq |\Var_{q,j} (t)  - \sigma_{q,j}^2 | \leq e^{-\gamma t} \lambda_j^2 c_\gamma. 
\end{equation*}
And similarly, we obtain
\[
\mathcal{W}_2^2\bigl(N(0, \Var_{v,j}(t)), \pi_{v,j}\bigr) \leq e^{-\gamma t} \lambda_j^2 c_\gamma. 
\]
The significance of the above estimates is, for each pair of Fourier coefficients, the convergence rate is independent of the mode index $j$. While by similar calculations, one can show that the Fourier coefficients of the \eqref{eq:Lang} system and the \eqref{eq:pLang} system have degenerate convergence rates as $j \rightarrow \infty$.

Thus, by summing yo all Fourier coefficients, we obtain the following convergence estimate for $\mathcal{L}(\mathfrak q(t), \mathfrak v(t))$,
\begin{equation}
\mathcal{W}_2(\mathcal{L}(\mathfrak q(t), \mathfrak v(t)), \Pi) \lesssim e^{-\frac{\gamma t}{2}} \sqrt{\sum_j \lambda_j^2} \lesssim e^{-\frac{\gamma t}{2}}.
\end{equation}

We can carry out similar calculations for $\gamma > 2$, when  the matrix $A$ has two distinct real negative eigenvalues, and the convergence estimate becomes
\begin{equation}
\mathcal{W}_2(\mathcal{L}(\mathfrak q(t), \mathfrak v(t)), \Pi) \lesssim e^{-\frac{(\gamma-\sqrt{\gamma^2-4}) t}{2}}.
\end{equation}

Finally, when  $\gamma = 2$, the matrix $A$ has two identical negative eigenvalues, and the the convergence estimate becomes
\begin{equation}
\mathcal{W}_2(\mathcal{L}(\mathfrak q(t), \mathfrak v(t)), \Pi) \lesssim e^{-t} t.
\end{equation}
This completes the proof.
\end{proof}

Theorem \ref{thm:est} establishes  the exponential convergence of the distribution of the SPDE \eqref{eq:pmmlspdeh} to its equilibrium measure, which implies that the finite dimensional system  \eqref{eq:pmmLang} converges to the invariant measure with a uniform rate. In contrast, for both  the \eqref{eq:Lang} system or the \eqref{eq:pLang} system, one can show by similar calculations that the convergence rate becomes degenerate for high modes.

It is also worth to comment on the dependence of the convergence rate on the damping parameter $\gamma$. When $\gamma \in (0,2)$, the convergence rate is increased as $\gamma$ increases. When $\gamma>2$, the convergence rate decreases as $\gamma$ increases. Observe also that when $\gamma \gg 1$, the exponent $\theta(t)$ is approximately $t/\gamma$, which is consistent with the fact that the Langevin dynamics converges to the overdamped limit under the time scaling $t \mapsto \gamma t$ as $\gamma \rightarrow \infty$.

Finally, we remark that the exponential convergence result in Theorem \ref{thm:est} is expected to be  generalized to a wider class of potential functions beyond the harmonic case.  The recent works by Zimmer \cite{Z17} and Bou-Rabee and Eberle \cite{EGZ19} proved geometric ergodicity of a family of  infinite-dimensional diffusions and the infinite dimensional preconditioned HMC based on a two-scale coupling approach. Similar ideas may be adopted to obtain a dimension-independent convergence rate for the preconditioned mass-modified  Langevin dynamics \eqref{eq:pmmlspd} with a general potential. We will investigate this problem in a future work.   

%\cite{BE19,EGZ19,Z17}

%\begin{align*}
%\Var_{q,j} (t) =& \frac{8 \gamma \lambda_j^2}{- \Delta} \int_0^t e^{- \gamma(t-s)} \sin^2  \left( \frac{\sqrt{-\Delta}(t-s)}{2} \right)  \ud s  \\
%= &  \lambda_j^2 + \lambda_j^2 e^{-\gamma t} \frac{4 \gamma}{-\Delta} \left( \frac{1}{- \gamma} - \frac{1}{4} \left( - \gamma \cos(\sqrt{-\Delta} t)+ \sqrt{-\Delta} \sin (\sqrt{-\Delta}t) \right) \right). 
%\end{align*}

\section{Preconditioned dynamics for multi-level systems} \label{sec:multi-level}

%\zz{this section has been reworked}

 In the diabatic representation, the Hamiltonian
operator of a general two-level system can be written as (atomic unit
is used)
\[ \wh H = \wh T+ \wh V = \frac{1}{2}
\begin{pmatrix} \wh p^2 & \\ & \wh p^2
\end{pmatrix} +
\begin{pmatrix} V_{00}(\wh q) & V_{01}(\wh q) \\ V_{10}(\wh q) &
V_{11} (\wh q)
\end{pmatrix},
\] where $\wh q$ and $\wh p$ are the nuclear position and momentum
operators, and the mass of the nuclei is chosen to be $1$ (again, for
simplicity we assume all nuclei have the same mass).  The multi-level
quantum systems arise when the non-adiabatic effect between different
energy surfaces of electronic states cannot be neglected, see e.g.,
the review articles \cite{Makri:99, StockThoss:05, Kapral:06}.

The thermal equilibrium average of
observables, given by
\begin{equation}\label{eq:aveA} \langle\wh{A}\rangle = \frac{\tr_{ne}
[e^{-\beta \wh H} \wh A ] }{\tr_{ne}[e^{-\beta \wh H}]},
\end{equation} for an operator $\wh A$, where $\beta = \frac{1}{k_B
T}$ with $k_B$ the Boltzmann constant and $T$ the absolute
temperature, and $\tr_{ne}$ denotes trace with respect to both the
nuclear and electronic degrees of freedom, namely,
\[ \tr_{ne}=\tr_{n}\tr_{e}=\tr_{L^2 (\R^d)}\tr_{\mathbb C^{2}}.
\] The denominator in \eqref{eq:aveA} is the partition function given
by $\mathcal{Z} = \tr_{ne}[e^{-\beta \wh H}]$.

In \cite{LZPIMD1,LZPIMD2}, it is shown that the thermal average can be viewed as (up to a
normalization) {an average with respect to the classical Gibbs
distribution} for {ring polymers on the extended configuration space}:
\begin{equation}\label{eq:ensembleavgA} \langle \wh{A} \rangle \approx \int_{\mathbb R^{2dN}} \ud \bd q \ud \bd p
\sum_{\bd{\ell}\in\{0, 1\}^N} \pi_N (\bd q,  \bd p, \bd \ell )W_N[A]
(\bd q,\bd \ell),
\end{equation} with the distribution
\begin{equation}\label{eq:pi} \pi_N (\bd q,  \bd p, \bd \ell ) =
\frac{1}{\mathcal Z_N'} \exp(-\beta_N H_N ({\bd q, \bd p,\bd \ell})).
\end{equation} 
Here, the extended phase space variable  $ (\bd{q}, \bd{p},
\bd{\ell}) \in \mathbb R^{2dN} \times \{0, 1\}^N$, where  $\bd{q} = (q_1, \cdots, q_N)$  are the position of each bead,  $\bd{p} = (p_1, \cdots, p_N)$  are the momentum of each bead and  $\bd{\ell} = (\ell_1, \cdots, \ell_N)$
indicates the energy level of the bead.  The expressions for $W_N[A]$ and $H_N$ can be found in
Appendix~\ref{app:two-level}, and the readers may also refer to
\cite{LZPIMD1} for detailed derivations.
Similar to the single-level case, the Hamiltonian can be rewritten as
\begin{equation}\label{eq:Ham2l}
  \begin{aligned}
    H_N (\bd q, \bd p, \bd \ell)  &=   \frac 1 2 \bd p \cdot M^{-1} \bd p + \frac 1
    2 \bd q \cdot L^{} \bd q + V_N (\bd q, \bd \ell) \\
    & = \frac 1 2 \bd p \cdot M^{-1} \bd p + \frac 1
    2 \bd q \cdot L^{\alpha} \bd q + U_N^\alpha (\bd q, \bd \ell).
\end{aligned}
\end{equation}
Note that, the dependence on $\bd{\ell}$ is all contained in the $V_N$
(or equivalently, $U_N^\alpha$ ) part of the Hamiltonian.  Ergodic
trajectories with respect to the distribution \eqref{eq:pi} have been
introduced in \cite{LZPIMD1} to effectively sample in thermal averages
in such two-level systems.  In \cite{LZPIMD2}, the infinite swapping
limit of the sampling path has been discussed to enhance the numerical
performance of the algorithm when the observables have off-diagonal
elements. It amounts to integrate out the fast degree of freedom
$\bd \ell$, which leads to the following averaged Hamiltonian
\begin{align}
  \wb{H}_N(\bd{q},\bd{p}) &= -\frac{1}{\beta_N} \ln \Bigl(\sum_{\bd{\ell} \in \{0, 1\}^N}  e^{-\beta_N H_N(\bd{q}, \bd{p}, \bd{\ell})} \Bigr)  \\
&=  \frac 1 2 \bd p \cdot M^{-1} \bd p + \frac 1
    2 \bd q \cdot L^{\alpha} \bd q  + \wb{U_N^{\alpha}}(\bd q). \nonumber
\end{align}
where $\wb{U_N^{\alpha}}(\bd q)$ denotes the averaged potential
\[
\wb{U_N^{\alpha}}(\bd q)=-\frac{1}{\beta_N} \ln \Bigl(\sum_{\bd{\ell} \in \{0, 1\}^N}  e^{-\beta_N U_N^{\alpha}(\bd q. \bd \ell)} \Bigr) 
\]
and correspondingly, the averaged distribution is given by
\begin{equation}
  \wb{\pi}(\bd{q}, \bd{p})  \propto \exp(-\beta_N \wb{H}_N(\bd{q}, \bd{p})).
\end{equation}
We define the conditional probability
distribution in $\bd \ell$ with fixed $(\bd{q}, \bd{p})$, namely, 
%\jl{perhaps we shall stick with $(\bd{p}, \bd{q})$ instead of introducing $\bd{z}$?}
\begin{equation}\label{eq:reZz}
\pi (\bd \ell \mid \bd{q}, \bd{p})= \frac{1}{\exp(-\beta_N \wb{H}_N(\bd{q}, \bd{p}) )} \exp(-\beta_N H_N (\bd{q}, \bd{p}, \bd \ell )).
\end{equation}
and the averaged observable 
\begin{equation} \label{eq:aWA}
\widetilde W [A](\bd{q}, \bd{p}) := \sum_{\bd{\ell} \in \{0, 1\}^N} \pi(\bd\ell \mid \bd{q}, \bd{p}) W_N[A](\bd{q}, \bd{p}, \bd \ell), 
\end{equation}
then  the thermal average is approximated by
\begin{equation} \label{eq:aveIS}
\langle \wh{A} \rangle  \approx  \int_{\RR^{2dN}} \ud \bd z \,   \wb{\pi}(\bd{z})  \widetilde W [A](\bd z).
\end{equation}

 If we can construct a trajectory $
{\bd{z}}(t)$ that is ergodic with respect to the equilibrium
distribution $\wb \pi$, we can sample the ensemble average on the right
hand side of \eqref{eq:aveIS} by a time average to approximate
$\langle\wh{A}\rangle$:
\begin{equation} \langle \wh{A} \rangle \approx \lim_{T \rightarrow
\infty } \frac{1}{T} \int_0^T \widetilde W_N [A] ( {\bd{z}} (t)) \ud t.
\end{equation} 

In \cite{LZPIMD2}, the infinite swapping limit of the PIMD-SH method was introduced to sample the thermal average \eqref{eq:aveIS}, where $ {\bd{z}}(t)$ evolves according to the following Langevin dynamics 
\begin{align} \label{eq:qa}
  {\ud \bd q}&=  M^{-1}{\bd p} \ud t. \\
  \label{eq:pa}
  {\ud \bd p}&= - \nabla_{\bd q} \wb{H}_N(\bd q, \bd p) \ud t - \gamma \bd p \ud t + \sqrt{2 \gamma \beta_N^{-1}M} \ud \bd B \\
& =  -L^\alpha  \bd q \ud t -\nabla_{\bd q}  \wb{U_N^{\alpha} } (\bd q)\ud t  \nonumber \\
& \quad \quad \quad - \gamma \bd p \ud t  + \sqrt{2 \gamma \beta_N^{-1}M} \ud \bd B. \nonumber
\end{align}
%\begin{equation}
%  \mc{Z}_{\bd{z}} = \exp(-\beta_N \wb{H}_N(\bd{q}, \bd{p}) ).
%\end{equation}
We observe that, compared with the single level case, the sampling dynamics only differs in the $-\nabla_{\bd q}  \wb{U_N^{\alpha} } (\bd q)$ term, and the preconditioning techniques introduced in Section~\ref{sec:underdamp} can be directly applied.
In the following, we only consider the preconditioning for the mass-modified Langevein system.

Similar to the single level case,  $M$ is a
positive definite fictitious mass matrix for the auxiliary momentum
variables $\bd{p}$.  We choose $M=L^\alpha$ and introduce the velocity variable $\bd v =M^{-1} \bd p$, then, we obtain the mass-modified Langevin dynamics
\begin{equation} \label{eq:pmmLang2}
\begin{aligned} 
  {\ud \bd q}&=  { \bd v} \ud t;\\
  {\ud  \bd v}&= -  \bd q \ud t - (L^\alpha)^{-1} \nabla_{\bd q} \wb{ U_N^{\alpha}} (\bd q) \ud t  - \gamma  \bd v\ud t  \\ & \quad \quad  \quad \quad + \sqrt{\frac{2 \gamma  (L^\alpha)^{-1}}{\beta_N}} \ud \bd B.
\end{aligned}
\end{equation}
And with this change of variables, the Hamiltonian \eqref{eq:Ham2l} is rewritten as
\begin{equation}\label{eq:Ham2l2}
  \begin{aligned}
    \wb{H'_N} (\bd q, \bd v)  &=   \frac 1 2 \bd v \cdot L^{\alpha} \bd v + \frac 1
    2 \bd q \cdot L^{\alpha} \bd q + \wb{U_N^\alpha} (\bd q).
\end{aligned}
\end{equation}

To conclude this section, we remark that the preconditioning
techniques, in theory,  are compatible with other numerical approaches in the
non-adiabatic regime, %\jl{do you mean adiabatic or non-adiabatic?} 
e.g. \cite{LL18,TSM18}. We shall omit the
derivations in this work, and only we will carry out numeric experiments
for the preconditioned mass-modified IS method (pmmIS). 

%\jl{I guess we
%  shall say some words on how this is done; and in fact why the
%  generalization is straightforward? I would perhaps also start from
%  the beginning say that the preconditioning idea can be applied both
%  to the normal and IS versions; and we will present the normal
%  version for simplicity of presentation and comment on its extension
%}

\section{Numerical Tests} \label{sec:num}

\subsection{BAOAB integrator for Langevin dynamics}
For numerical integration of the sampling dynamics, we will use
variants of the BAOAB scheme \cite{LeimkuhlerMatthews} which is widely
used for the original Langevin dynamics \eqref{eq:Lang}. We discuss
the adaptations of the BAOAB scheme to the proposed systems. 

Let us consider a general Langevin dynamics as
\begin{align} \label{eq:qagen}
  {\ud \bd q}&=  {C_1 M^{-1} \bd p} \ud t;\\
  {\ud  \bd p}&= -C_1 L^\alpha \bd q \ud t -C_1\nabla_{\bd q} U^\alpha_N \ud t   \label{eq:pagen} \\
& \quad \quad - \gamma  C_2 \bd p\ud t + \sqrt{\frac{2 \gamma   C_2 M}{\beta_N}} \ud \bd B,  \nonumber
\end{align}
where $C_1$ and $C_2$ are some positive definite matrices, and we
assume $C_2$ and $M$ commute: $C_2M=MC_2$.  Obviously, when
$C_1=C_2=I$ and $M= m I$, the system \eqref{eq:qagen}-\eqref{eq:pagen}
reduces to \eqref{eq:Lang}. When $C_1=C_2=(L^\alpha)^{-1}$ and
$M= m I$, the system \eqref{eq:qagen}-\eqref{eq:pagen} becomes
\eqref{eq:pLang}. It reduces to \eqref{eq:mmLang} when $C_1=C_2=I$ and
$M= L^\alpha$. The original BAOAB scheme easily extends to the system
\eqref{eq:pmmLang} in $\bd q$ and $\bd v$ variables, which we will
skip the details.

BAOAB scheme is based on operator splitting, where the whole dynamics
is divided into the kinetic part (denoted by ``A'')
\begin{align} 
  {\ud \bd q}&=  C_1 M^{-1}{ \bd p} \ud t;\\
  {\ud  \bd p}&= \bd 0,
\end{align}
the potential part (denoted by ``B'')
\begin{align} 
  {\ud \bd q}&=\bd 0;\\
  {\ud  \bd p}&=  C_1 L^\alpha \bd q \ud t - C_1 \nabla U_N^{\alpha}( \bd q)-  \ud t
\end{align}
and the thermostat part (denoted by ``O'')
\begin{align}
  {\ud \bd q}&= \bd  0;\\
  {\ud  \bd p}&= -  \gamma C_2  \bd p \ud t + \sqrt{2 C_2 \gamma \beta_N^{-1} M} \ud  \bd B.
\end{align}
Note that, the potential step and the kinetic step can be solved analytically and the thermostat part allows the following exact solution in the sense of distributions
\begin{align}
  {\bd q}(t) & =  \bd q(0);\\
  {  \bd p} (t)& =e^{ -  \gamma C_2 t} \bd p(0)  + \sqrt{(1-e^{2 \gamma  C_2 t})(\beta_N^{-1} M)} \,\, \bd \xi,
\end{align}
where $\bd \xi$ is $dN$ dimensional, with each component an
independent standard Gaussian random variable.  We also remark that if
time step size $\Delta t$ is fixed, the matrix $e^{ - \gamma C_2 t}$
and $\sqrt{1-e^{2 \gamma C_2 t}}$ can be precomputed. Thus, in each
time evolution step, the computation complexity is dominated by the
matrix-vector multiplication.

\subsection{Numerical examples}

To compare the performance of the diagonal-mass Langevin dynamics
\eqref{eq:Lang}, the preconditioned diagonal-mass Langevin dynamics
\eqref{eq:pLang}, the mass-modified Langevin dynamics
\eqref{eq:mmLang} and the preconditioned mass-modified Langevin
dynamics \eqref{eq:pmmLang}, we carry out numerical tests using the
following two examples, in one and two dimension, with computational
domain with periodic boundary conditions. The reference solutions are
obtained with pseudo-spectral discretization (which is possible thanks
to the low-dimensionality). In these examples, a large number of beads
are needed to reduce the asymptotic error, and thus manifest
differences of the various Langevin dynamics introduced above.

In the 1D test problem, the potential function is given by
\begin{equation}\label{ex:pot1}
V(q)=10-10\cos\bigl(q\bigr)+5\cos\bigl(2(q-0.1)\bigr).
\end{equation}
This potential surface is plotted in Figure \ref{fig:Eplot1}, from
which, we observe that the potential function has two local minima
around $x=\pm 1$ respectively, and the barrier separating the two
minima is located around $x=0$. We choose inverse temperature
$\beta=8$ and observable is given by
\begin{equation}\label{ex:ob1}
A(q)=e^{-10q^2}.
\end{equation}
Since the observable is localized around the local maximum of the potential and the prescribed temperature is fairly low compared to the potential barrier, many beads are needed to represent the ring polymer configuration that extending from the local minima to the saddle point. 
\begin{figure}[ht]
\begin{centering}
\includegraphics[scale=0.55]{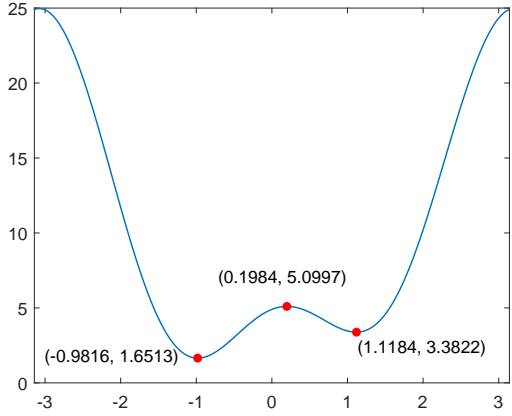} \\
\caption{Potential surface of the 1D test example \eqref{ex:pot1}. The
  red dots mark local maximum and local minima.}
\label{fig:Eplot1}
\end{centering}
\end{figure}

In the 2D example, the potential function is given by a three-well model
in 2D (see Fig.~\ref{fig:Eplot2}):
\begin{equation}\label{ex:pot2}
  \begin{aligned}
V(q_1,q_2)=& \,\,12-3\bigl(1+\cos(q_1) \bigr) \bigl(1+\cos(q_2) \bigr)
\\
& +3 e^{-5 q_1^2-5(q_2-0.2)^2}- 3 e^{-5 q_1^2-5(q_2-0.6)^2}\\
& -5 e^{-5 (q_1-0.6)^2-5 q_2^2}- 5 e^{-5 (q_1+0.6)^2-5 q_2^2}.
\end{aligned}
\end{equation}
Around the origin, there are two deeper wells located around
$(\pm 0.6,0)$ and a shallower well located around $(0,0.6)$.  The
inverse temperature is given by $\beta=8$ and the test observable is a
Gaussian located around the shallower well
\begin{equation}\label{ex:ob2}
A(q_1,q_2)=e^{-10\left(q_1^2+(q_2-0.6)^2\right)}.
\end{equation}
Thus a large number of beads are needed to represent the ring polymer configuration. 

\begin{figure}[ht]
\begin{centering}
\includegraphics[scale=0.55]{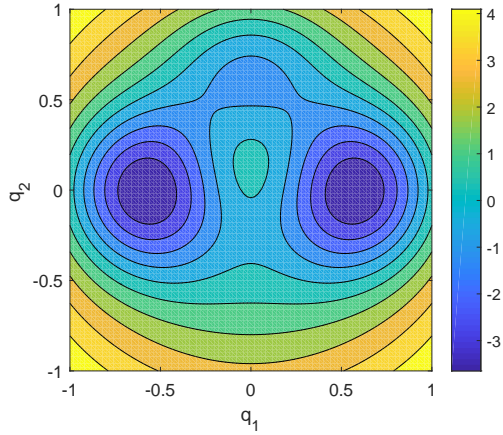} \\
\caption{Potential surface of the 2D test example
  \eqref{ex:pot2}. 
%The red stars mark the approximate well
 % bottoms. \jl{I don't see the red stars?} 
At $(\pm 0.6,0)$, the
  potential takes the value $-3.6317$, and at $(0,0.6)$, the potential
  takes the value $ -0.8773$.}
\label{fig:Eplot2}
\end{centering}
\end{figure}

{
For the two-level example,  the test potential is %\jl{do we ever choose $\delta$ that is not $1$?}
\begin{equation}  \label{ex:pot1-2}
\left \{
\begin{split}
V_{00} &=a \bigl(1-\cos(x) \bigr); \\
V_{11} & = b\bigl(1-\cos(x) \bigr);\\
V_{01} & = V_{10} =  c e^{-d x^2}.
\end{split}
\right.
\end{equation}
We take $b>a$, so $V_{11} \ge V_{00}$ and the two energy surfaces only
intersect at $x=0$, where the off-diagonal term takes its largest
value.  %The parameter $\delta$ determines the coupling strength.
The energy surfaces are symmetric with respect to $x = 0$. At thermal
equilibrium, the density is expected to concentrate around $x=0$,
where transition between the two surfaces is the most noticeable due
to the larger off-diagonal coupling terms. In this work, we
choose $a=4$, $b=8$, $c=1$ and $d=1$. We plot the diabatic energy surfaces in Figure~\ref{fig:Eplot2}. In this example, 
we test and compare the performances of
numerical methods with the diagonal observable
\begin{equation} \label{eq:ob1}
\wh A=\begin{bmatrix}
  e^{- {\wh q}^2}   &  0\\
 0 &   e^{- {\wh q}^2}  
\end{bmatrix},
\end{equation}
and also the off-diagonal observable
\begin{equation} \label{eq:ob2}
\wh{A}=\begin{bmatrix}
0   &    e^{- {\wh q}^2}\\
  e^{- {\wh q}^2} &   0  
\end{bmatrix}.
\end{equation}
}

\begin{figure}[htbp]
\begin{center}
\includegraphics[scale=0.55]{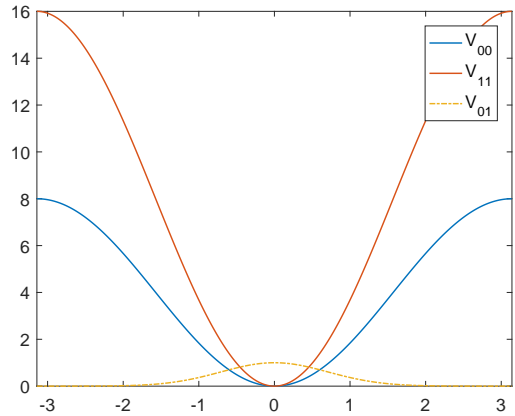} \\
\caption{Diabatic potential surfaces for the test example
  \eqref{ex:pot1-2}. }
\label{fig:Eplot2}
\end{center}
\end{figure}

\subsection{Stability and convergence tests with different time step sizes} \label{sec:stability}

In this section, we focus on the 1D test example, with the potential
function and the observable given by \eqref{ex:pot1} and
\eqref{ex:ob1}, respectively, and we choose the inverse temperature
$\beta=8$. In the \eqref{eq:Lang} system and \eqref{eq:pLang} system,
we choose the scalar mass $m=1$. We first aim to test the four
Langevin dynamics \eqref{eq:Lang}, \eqref{eq:pLang}, \eqref{eq:mmLang}
and \eqref{eq:pmmLang} using the BAOAB scheme with various time steps
and numbers of beads. In particular the time step size restrictions
when the number of beads is large. We further study the convergence
behavior of the PIMD simulations with respect to the time step size.

We test the BAOAB method with $\Delta t=1$, $\frac 1 4$,
$\frac 1 {16}$, $\frac 1 {64}$ and $\frac 1 {256}$. The errors in the
empirical averages till simulation time $T=10,000$ from each
simulation are reported in Table~\ref{table:test1_1} and
Table~\ref{table:test1_3}.  We observe that, the two preconditioned
Langevin dynamics \eqref{eq:pLang} and \eqref{eq:pmmLang} show
superior numerical stability compared to the other two systems, as
expected. The numerical results are stable even for $\Delta t=1$, and
they provide accurate approximation of the observable when
$\Delta t<1$.
\begin{table}[pth]
  \centering
  \begin{tabular}{ c| c| c|c|c|c} \toprule
    \# beads  & $\Delta t=1$ & $\Delta t=\frac{1}{4}$ & $\Delta t=\frac{1}{16}$ & $\Delta t=\frac{1}{64}$ & $\Delta t=\frac{1}{256}$  \\ \hline
    32& NaN & 9.74e-2 & 1.15e-2 & 1.05e-2 & 9.67e-3
    \\ \hline
    64 & NaN & NaN & 2.30e-3 & 3.11e-3 & 1.31e-3 
    \\ \hline
    128 & NaN & NaN & 9.79e-2 & 3.85e-4 & 4.03e-4  \\ \bottomrule
  \end{tabular}
  \medskip
  
  \begin{tabular}{ c| c| c|c|c|c} \toprule
    \# beads  & $\Delta t={1}$ & $\Delta t=\frac{1}{4}$ & $\Delta t=\frac{1}{16}$ & $\Delta t=\frac{1}{64}$ & $\Delta t=\frac{1}{256}$  \\ \hline
    32& 9.86e-2 & 1.07e-2 & 1.22e-2 & 1.10e-2 & 7.01e-3
    \\ \hline
    64 & 9.80e-2 & 4.10e-3 & 3.01e-3 & 3.75e-3 & 4.69e-3 
    \\ \hline
    128 & 9.87e-2 & 2.88e-3 & 2.73e-3 & 4.70e-4 & 2.15e-4  \\ \bottomrule
\end{tabular}
\caption{1D Example. Numerical empirical averages computed with various time step sizes and various numbers of beads. The reference value is 9.8734e-2.  Top: the \eqref{eq:Lang} system. Bottom: the \eqref{eq:pLang} system. "NaN" means the numerical integrator is unstable.
%\jl{why don't report the error instead? which will be much more clear to read. Also explain NaN (perhaps replace it with   something else says unstable); use equation reference}
%  \zz{Please note that in the two tables, I still report the average value instead of the corresponding errors because I did not really want to show quantitative results in these tables yet. But let me know if you think reporting errors is a better choice.} \jl{i still prefer reporting the error, which is far more clear.}
}
  \label{table:test1_1}
\end{table}

\begin{table}[pthb]
  \centering
  \begin{tabular}{ c| c| c|c|c|c} \toprule
   \# beads  & $\Delta t={1}$ & $\Delta t=\frac{1}{4}$ & $\Delta t=\frac{1}{16}$ & $\Delta t=\frac{1}{64}$ & $\Delta t=\frac{1}{256}$  \\ \hline
 32& NaN & NaN & NaN & 1.16e-2 & 9.62e-3
    \\ \hline
   64 & NaN & NaN & NaN & NaN & 2.91e-3 
    \\ \hline
128 & NaN & NaN & NaN & NaN & NaN  \\ \bottomrule
\end{tabular}
\medskip

  \begin{tabular}{ c| c| c|c|c|c} \toprule
   \# beads  & $\Delta t={1}$ & $\Delta t=\frac{1}{4}$ & $\Delta t=\frac{1}{16}$ & $\Delta t=\frac{1}{64}$ & $\Delta t=\frac{1}{256}$  \\ \hline
 32& 9.75e-2 & 1.03e-2 & 9.94e-3 & 1.04e-2 & 9.09e-3
    \\ \hline
   64 & 9.84e-2 & 2.83e-3 & 3.00e-3 & 3.21e-3 & 1.05e-3 
    \\ \hline
128 & 9.82e-2 & 1.03e-3 & 2.48e-4 & 1.33e-3 & 4.23e-4  \\ \bottomrule
\end{tabular}
\caption{1D Example. Numerical empirical averages computed with various time step sizes and various numbers of beads. The reference value is 9.8734e-2. Top:  the \eqref{eq:mmLang} system, Bottom: the \eqref{eq:pmmLang} system. "NaN" means the numerical integrator is unstable.}
  \label{table:test1_3}
\end{table}

In comparison, the numerical solution for \eqref{eq:Lang} blow up for
large numbers of beads or large time step sizes, caused by
instability. The stability constraint is even more severe for
\eqref{eq:mmLang} We further observe that in the \eqref{eq:Lang}
system and the \eqref{eq:mmLang} system, as the number of beads
increases, one needs to take smaller time steps in integration the
sampling trajectories for the sake of stability.

Finally, we observe from the tables that, when the number of bead
equals $128$, the numerical results are closer to the reference
values. It confirms the intention of designing such an example, due to
the large potential barrier and the low temperature, many beads are
needed to reduce the asymptotic error in the ring polymer
approximation.

%\begin{table}
%  \centering
%  \begin{tabular}{ c| c| c|c|c|c} \hline
%   \# beads  & $\Delta t={1}$ & $\Delta t=\frac{1}{4}$ & $\Delta t=\frac{1}{16}$ & $\Delta t=\frac{1}{64}$ & $\Delta t=\frac{1}{256}$  \\ \hline
% 32& 1.27e-3 & 8.85e-2 & 8.88e-2 & 8.83e-2 & 8.96e-2
%    \\ \hline
%   64 & 3.36e-4 & 9.59e-2 & 9.57e-2 & 9.55e-2 & 9.76e-2 
%    \\ \hline
%128 & 4.94e-4 & 9.77e-2 & 9.87e-2 & 9.74e-2 & 9.83e-2  \\ \hline
%\end{tabular}
%\caption{Numerical empirical averages computed from the LLP system with various time step sizes and various numbers of beads. The reference value is 9.8734e-2. }
%  \label{table:test1_4}
%\end{table}

\medskip 

Next, we test the two preconditioned dynamics \eqref{eq:pLang} and
\eqref{eq:pmmLang} for convergence with respect to the time step
sizes. To reduce the effect of the asymptotic error in the ring
polymer approximation, we take the number of beads $N=128$. We test
the two systems with $\Delta t= \frac 1 2$, $\frac 1 4$, $\frac 1 8$
and $\frac 1 {16}$ with simulation time $T=40,000$. We plot the
running averages of the PIMD simulation for \eqref{eq:pLang} and
\eqref{eq:pmmLang} in Figure \ref{fig:time1}. We observe the the plots
that, while some estimation bias is present for $\Delta t= \frac 1 2$,
it is becomes not unnoticeable when $\Delta t$ takes smaller values.

\begin{figure}[ht]
\begin{centering}
\includegraphics[scale=0.5]{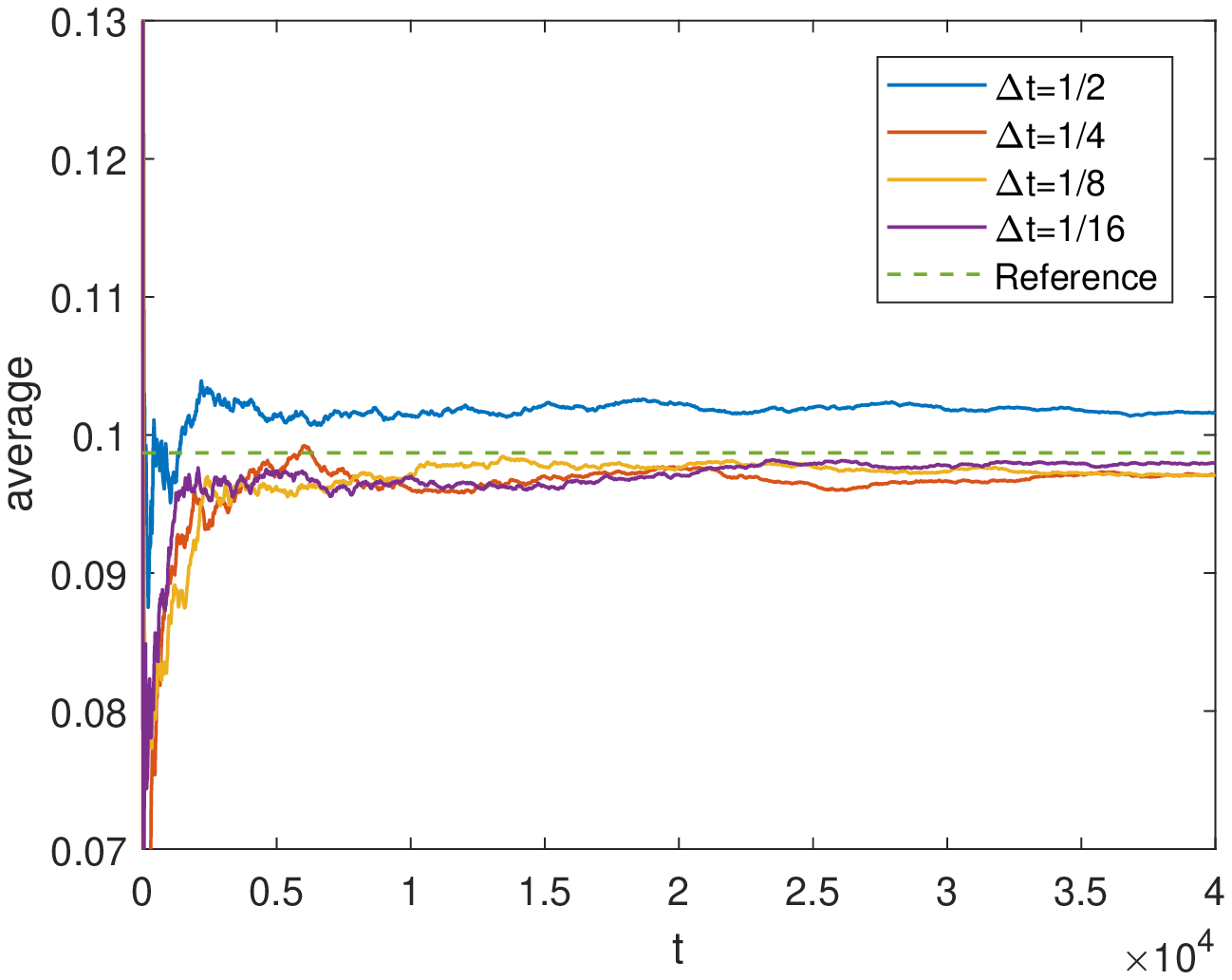} 
\includegraphics[scale=0.5]{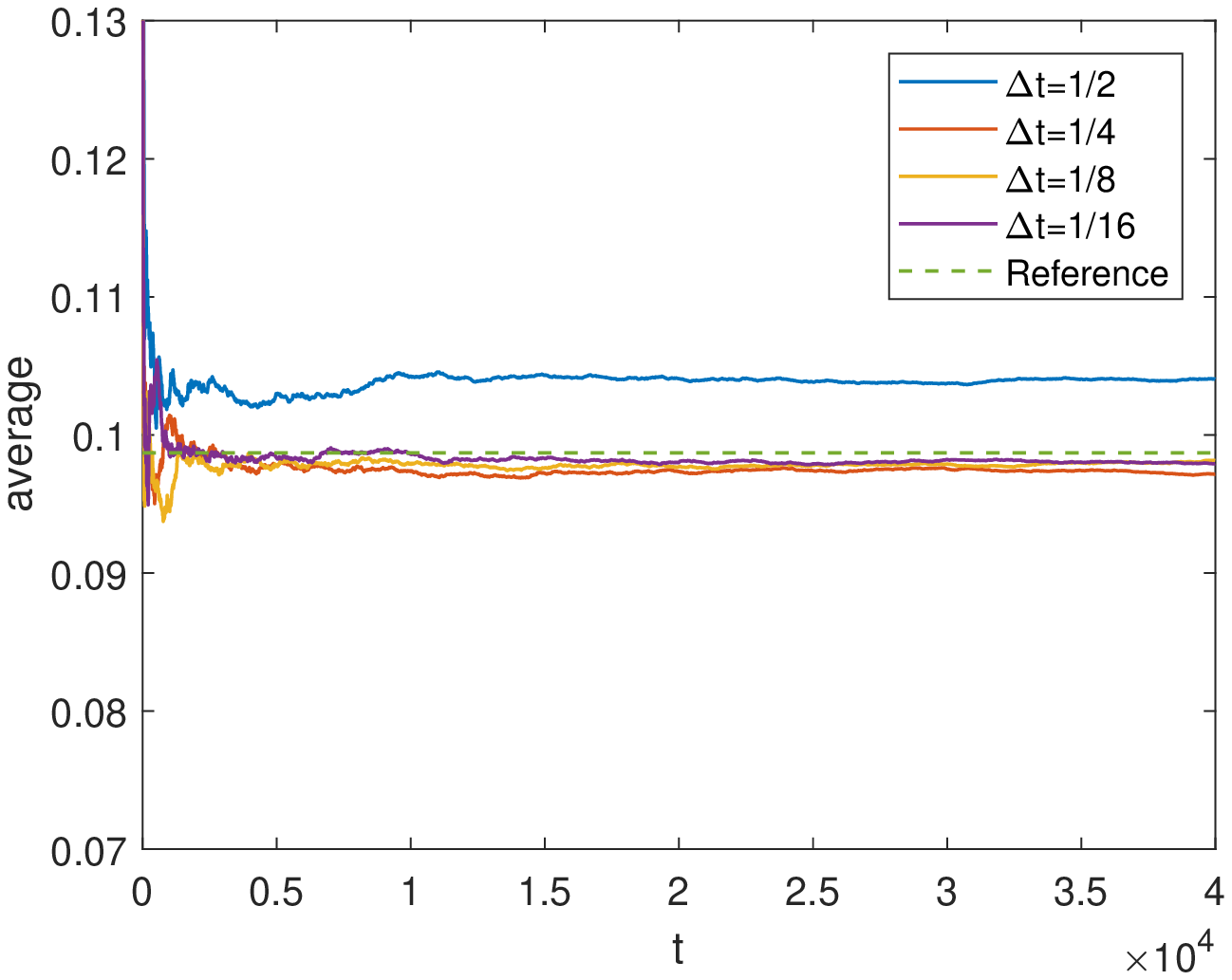} \\
\caption{1D Example. Running averages of the PIMD method with various
  time steps. The reference value is $9.8734e$-$2$.  Left: the
  \eqref{eq:pLang} system, Right: the \eqref{eq:pmmLang} system.}
\label{fig:time1}
\end{centering}
\end{figure}
%\ylcomment{In the caption of the figure, use "Left and Right" instead of "Top and Bottom". Of course, this may depend on the journal format.}

%\begin{figure}[ht]
%\begin{centering}
%\includegraphics[scale=0.55]{time2.eps} \\
%\caption{Running averages of the PIMD method applied to the LLP system with various time steps. The reference value is 9.8734e-2. }
%\label{fig:time2}
%\end{centering}
%\end{figure}

We can further confirm our observation by looking at the mean squared
error as a function of sampling time for those two systems as in
Figure~\ref{fig:mse1}. When $\Delta t = \frac 1 2$, we observe in
either case the sampling error is saturated around $t=1,000$, and when
$\Delta t = \frac 1 4$ or $\frac 1 8$, the sampling error is saturated
around $t=10,000$. This implies that when the simulation time is long
enough, the bias introduced by numerically integrating the sampling
trajectories with large time steps dominates the mean squared error.

The numerical results suggest we can take $o(1)$ time steps for
accurate approximation of the ensemble average, even for very large
number of beads. From this perspective, the \eqref{eq:pLang} system
and the \eqref{eq:pmmLang} system are better platforms for PIMD
simulations, because we only need to take small time steps for
accuracy, but not for stability constraints. 

\begin{figure}[ht]
\centering
\includegraphics[scale=0.5]{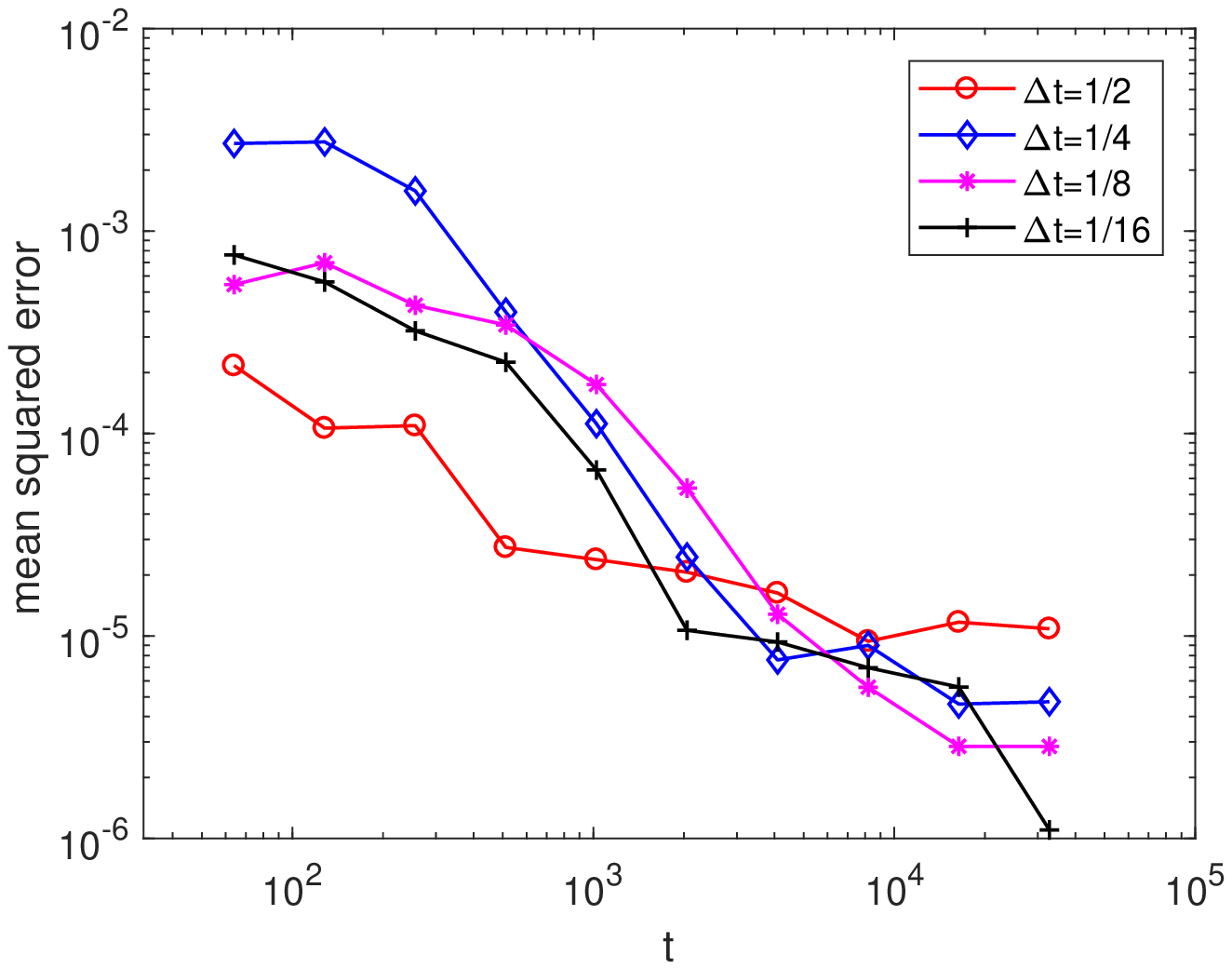} 
\includegraphics[scale=0.5]{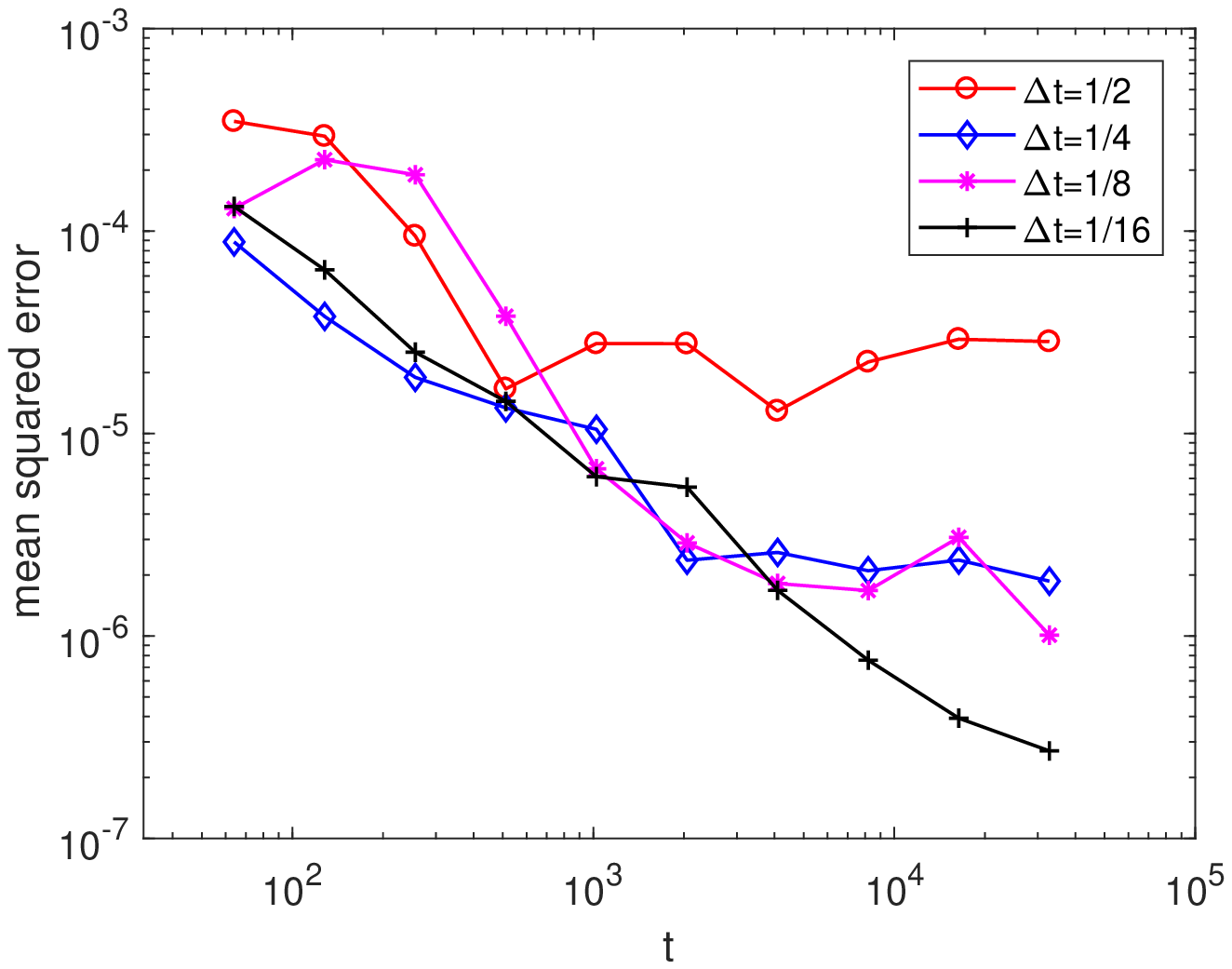} \\
\caption{1D Example. Mean squared errors at $t=2^6,\, 2^7, \cdots, 2^{15}$ by the BAOAB method with various time step sizes. Left:  the \eqref{eq:pLang} system. Right: the \eqref{eq:pmmLang} system.}
\label{fig:mse1}
\end{figure}

%
%\begin{figure}[ht]
%\begin{centering}
%\includegraphics[scale=0.55]{msellp.eps} \\
%\caption{Mean squared errors at $t=2^6,\, 2^7, \cdots, 2^{15}$ by the BAOAB method with various time step sizes applied to the LLP system.}
%\label{fig:mse2}
%\end{centering}
%\end{figure}

\subsection{Convergence with respect to the number of beads}

In this part, we aim to compare the PIMD simulation for the
\eqref{eq:pLang} system and the \eqref{eq:pmmLang} system from another
perspective.  We take $\Delta t= \frac 1 {50}$ to make sure the bias
due to numerical integration of the sampling trajectory is negligible
when the simulation time $T=10,000$.  The mean squared errors for
simulating those two systems are plotted in Figure \ref{fig:mse3}. We
observe in either test noticeable asymptotic bias when the number of
beads is $32$ or $64$, while the mean squared error decays in inverse
proportion to the simulation time when the number of beads equals
$128$.

\begin{figure}[ht]
\begin{centering}
\includegraphics[scale=0.5]{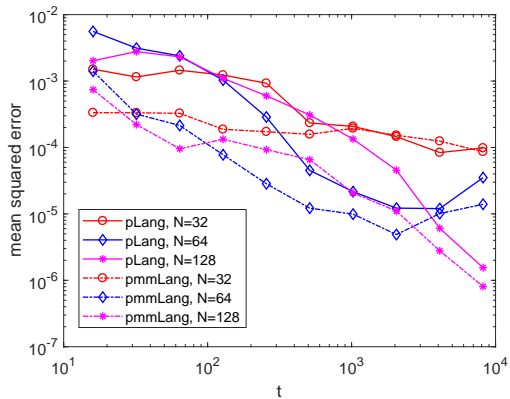} \\
\caption{1D Example. Mean squared errors at $t=2^4,\, 2^5, \cdots, 2^{13}$ by the BAOAB method with various number of beads.  Solid line: the \eqref{eq:pLang} system. Dash-dot line: the \eqref{eq:pmmLang} system.}
\label{fig:mse3}
\end{centering}
\end{figure}

%
%\begin{figure}[ht]
%\begin{centering}
%\includegraphics[scale=0.55]{mse4.eps} \\
%\caption{Mean squared errors at $t=2^4,\, 2^5, \cdots, 2^{13}$ by the BAOAB method with various number of beads applied to the LLP system.}
%\label{fig:mse4}
%\end{centering}
%\end{figure}

Moreover, we can observe some differences between the two systems
\eqref{eq:pLang} and \eqref{eq:pmmLang}. By comparing Figure
\ref{fig:mse1}, and Figure \ref{fig:mse3}, we find that when the
numerical error is dominated by the sample variance (small $\Delta t$
and large number of beads), sampling based on \eqref{eq:pmmLang} has
better accuracy. To understand better this observation, we plot in
Figure~\ref{fig:qauto} the autocorrelation of the $\bd q$
variable %(compute autocorrelation for each component of $\bd q$ and take the average)
for the number of beads $N=64$, $128$ and the time step
$\Delta t=\frac 1 {50}$. We clearly see that the autocorrelation time
of \eqref{eq:pmmLang} is much smaller than that of \eqref{eq:pLang}.
Thus \eqref{eq:pmmLang} produces more effective independent samples
with the same amount of simulation time. This is consistent with the
presence of slow modes in \eqref{eq:pLang} discussed above; and
\eqref{eq:pmmLang} dynamics converges to equilibrium faster.  We also
show in Table~\ref{table:test2_2} the correspondingly empirical error,
95\% confidence interval and the mean squared error at simulation time
$T=10,000$, which further verify that asymptotic sample variance of
the \eqref{eq:pmmLang} dynamics is indeed smaller.

From all the tests above, we conclude that the preconditioning as in
\eqref{eq:pLang} and \eqref{eq:pmmLang} improves PIMD sampling, while
the \eqref{eq:pmmLang} sampling dynamics has superior performance as
it reduces further the asymptotic variance.

\begin{figure}[ht]
\begin{centering}
\includegraphics[scale=0.5]{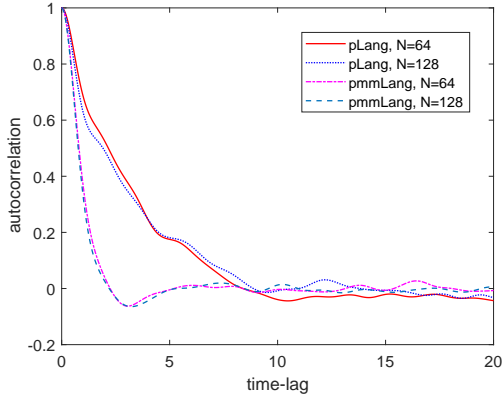} \\
\caption{1D Example. Autocorrelation of the position variables sampled from the \eqref{eq:pLang} system and the \eqref{eq:pmmLang} system.}
\label{fig:qauto}
\end{centering}
\end{figure}

\begin{table}[htpb]
  \centering
  \begin{tabular}{ c| c| c|c|c} \toprule
    & \multicolumn{2}{c|}{\eqref{eq:pLang} dynamics} & \multicolumn{2}{c}{\eqref{eq:pmmLang} dynamics} 
    \\    \hline
    & \,  N=64\, & \, N=128 \, & \, N=64 \, &  \, N=128\,  \\ 
    \hline
Error & 5.23e-3 & 9.78e-5 &  3.34e-3 & 4.02e-4 
    \\ \hline
   95\% C.I. & 2.91e-3 & 2.05e-3 & 1.44e-3 & 1.52e-3 
    \\ \hline
M.S.E. & 2.94e-5 & 1.10e-6 & 1.23e-5 & 7.70e-7   \\ 
    \bottomrule
\end{tabular}
\caption{1D Example. Errors in numerical empirical averages with $95$\% confidence intervals and mean squared errors. The reference value is {9.8734e-2.}}
  \label{table:test2_2}
\end{table}

\subsection{Comparison with the staging PIMD}

{To compare with the staging PIMD method
  (\cite{TBMK93,CM93,LiuLiLiu:16}), we repeat the 1D test problem with
  Hamiltonian in the staging coordinate \eqref{eq:Ham_staging}. Note
  that, to compare with the Langevin dynamics and the preconditioned
  versions proposed in Section~\ref{sec:underdamp}, we also applied
  the Langevin thermostat to \eqref{eq:Ham_staging} with inverse
  temperature $\beta_N$. We remark that, the staging PIMD that we test
  in the following is different from the one in \cite{TBMK93} in terms
  of thermostatting methods, but is similar to the Langevin dynamics
  in \cite{LiuLiLiu:16} although they chose the effective inverse
  temperature $\beta$ instead and they proposed the use of optimal
  friction coefficient.  In the first set of tests below, we stick to the choice the friction constant $\gamma=1$. 
  In fact, we have also repeated the tests in the second set, by the staging PIMD with parameters given in \cite{LiuLiLiu:16}.} 

%\jl{when you compare, what friction
%  coefficient is used?} \jl{perhaps we shall already mention that we
%  tested two version of staging transformation?}

{ We test the BAOAB method with $\Delta t=1$,
  $\frac 1 4$, $\frac 1 {16}$, $\frac 1 {64}$ and $\frac 1 {256}$. The
  errors in the empirical averages till simulation time $T=10,000$
  from each simulation are reported in Table~\ref{table:test1_5}. We
  observe that, with the staging coordinates, the numerical
  simulations are all stable, and the errors are similar to Langevin
  dynamics in Cartesian coordinates when the time steps are small
  enough. {Notice that, comparing with Table~\ref{table:test1_1} and Table~\ref{table:test1_3}, the numerical results seem to suggest that, when $\beta_N$ is small, the staging PIMD
  needs smaller time steps to show satisfactory accuracy.} 
%\jl{the previous   sentence seems to have issues} 
 We also checked autocorrelation of
  the position variables, which are mapped from the staging variables,
  and the autocorrelation show smaller scaled oscillations as the
  number of beads increases, as shown in
  Figure~\ref{fig:qauto_s1}. This observation seems to agree with the
  multi-scale behaviors of the staging PIMD dynamics as
  $\beta_N \ll 1$. We admit the the staging PIMD algorithm we test
  above is not optimized due to the freedom of specifying
  parameters. }

{ Recently, the authors in \cite{LiuLiLiu:16} proposed
  an optimized version of the staging PIMD. We also tested the problem
  with the setup in \cite{LiuLiLiu:16} and we observe that the
  numerical performances (in Table~\ref{table:test1_6} and
  Figure~\ref{fig:qauto_s1}) are improved for some cases. Comparing
  Table~\ref{table:test1_5} and Table~\ref{table:test1_6}, when the
  beads number is large, the time steps in the latter case which are
  needed to obtain accurate approximations are less restricted.  To
  sum up, the numerical tests seem to suggest that the Langevin
  dynamics in staging coordinates exhibit improved stability
  conditions and show comparable accuracy with the Langevin dynamics
  in Catesian coordinate when time steps are sufficiently
  resolved. Among all the sampling dynamics that we have tested, the
  \eqref{eq:pLang} system and the \eqref{eq:pmmLang} system are still
  favored because they both give the most relaxed constraints for the
  time steps to produce accurate simulation results, and the
  \eqref{eq:pmmLang} system is slightly superior as shown in the tests
  above.  For example, when $\Delta t = \frac 1 4$, both the
  \eqref{eq:pLang} system and the \eqref{eq:pmmLang} system already
  give relative errors around or lower than $10\%$, but in this case,
  the relative errors by staging PIMD are still beyond $50\%$. On the
  other hand, when the $\Delta t$ is chosen small enough, the
  numerical performance of the staging PIMD and the two preconditioned
  versions proposed are rather similar.}

%\jl{which result of the staging coordinate and of pmmLang are compared? It seems hard to conclude that pmmLang has better result from the table? We might need to be more precise here} 

\begin{table}[pthb]
  \centering
    \begin{tabular}{ c| c| c|c|c|c} \toprule
   \# beads  & $\Delta t={1}$ & $\Delta t=\frac{1}{4}$ & $\Delta t=\frac{1}{16}$ & $\Delta t=\frac{1}{64}$ & $\Delta t=\frac{1}{256}$  \\ \hline
 32& 4.90e-2 & 5.22e-2 & 2.61e-2 & 1.12e-2 & 9.57e-3
    \\ \hline
 64 & 5.93e-2 & 5.99e-2 & 7.29e-2 & 4.45e-3 & 2.55e-3 
    \\ \hline
128 & 6.91e-2 & 6.92e-2 & 8.39e-2 & 6.26e-2 & 1.96e-4  \\ \bottomrule
\end{tabular}
\caption{1D Example. Numerical empirical averages computed with various time step sizes and various numbers of beads. The reference value is 9.8734e-2. Sampling is obtained by the Langevin dynamics with staging coordinates.}
  \label{table:test1_5}
\end{table}

\begin{table}[pthb]
  \centering
  \begin{tabular}{ c| c| c|c|c|c} \toprule
   \# beads  & $\Delta t={1}$ & $\Delta t=\frac{1}{4}$ & $\Delta t=\frac{1}{16}$ & $\Delta t=\frac{1}{64}$ & $\Delta t=\frac{1}{256}$  \\ \hline
 32& 4.85e-2 & 9.17e-2 & 9.88e-3 & 9.93e-3 & 9.10e-3
    \\ \hline
 64 & 9.57e-2 & 7.29e-2 & 4.02e-3 & 2.46e-3 & 1.51e-3 
    \\ \hline
128 & 9.48e-2 & 9.65e-2 & 1.23e-3 & 6.16e-4 & 5.84e-4  \\ \bottomrule
\end{tabular}
\caption{1D Example. Staging PIMD with the setup~\cite{LiuLiLiu:16}. Numerical empirical averages computed with various time step sizes and various numbers of beads. The reference value is 9.8734e-2. Sampling is obtained by the Langevin dynamics with staging coordinates.}
  \label{table:test1_6}
\end{table}

\begin{figure}[ht]
\begin{centering}
\includegraphics[scale=0.55]{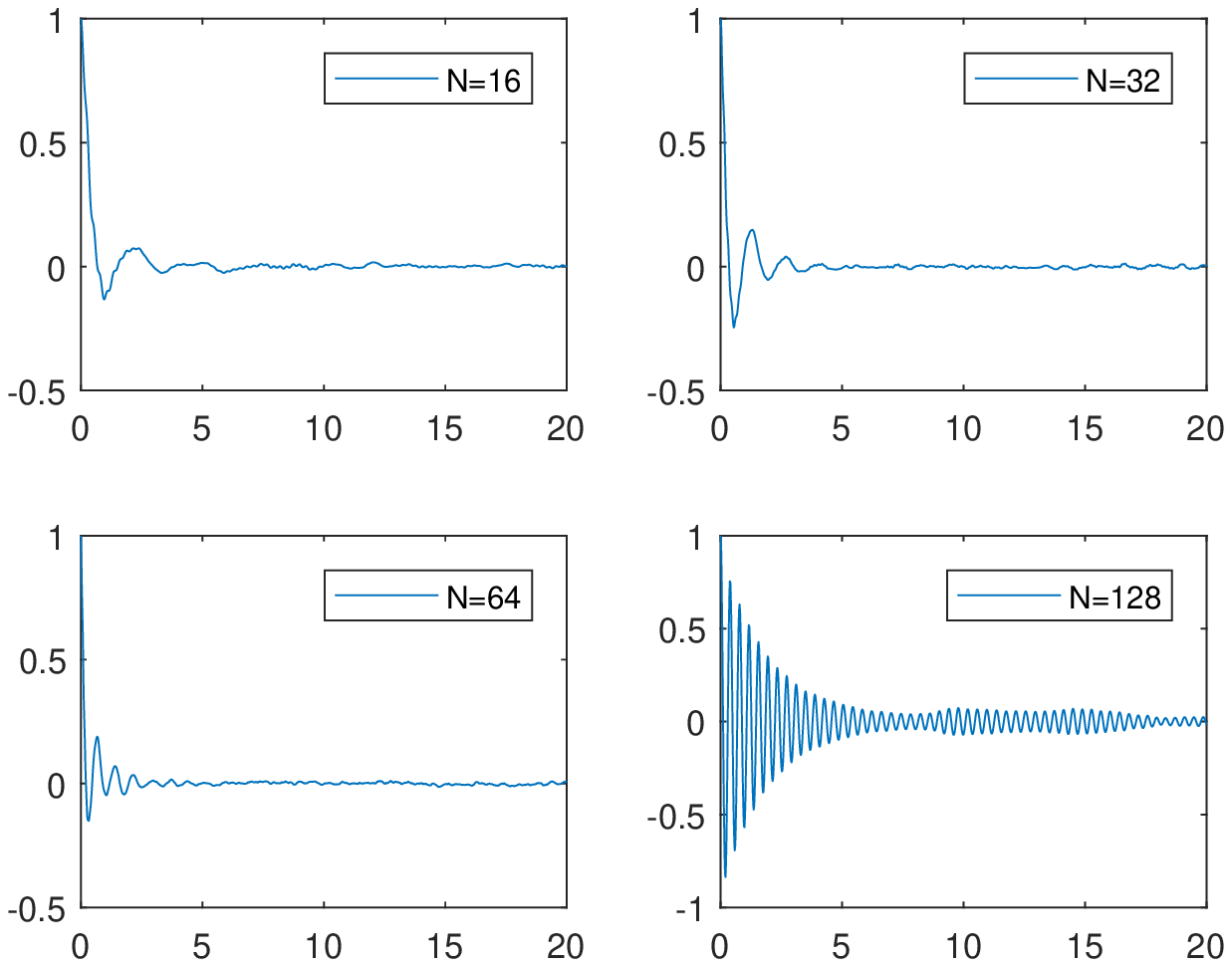}
\includegraphics[scale=0.55]{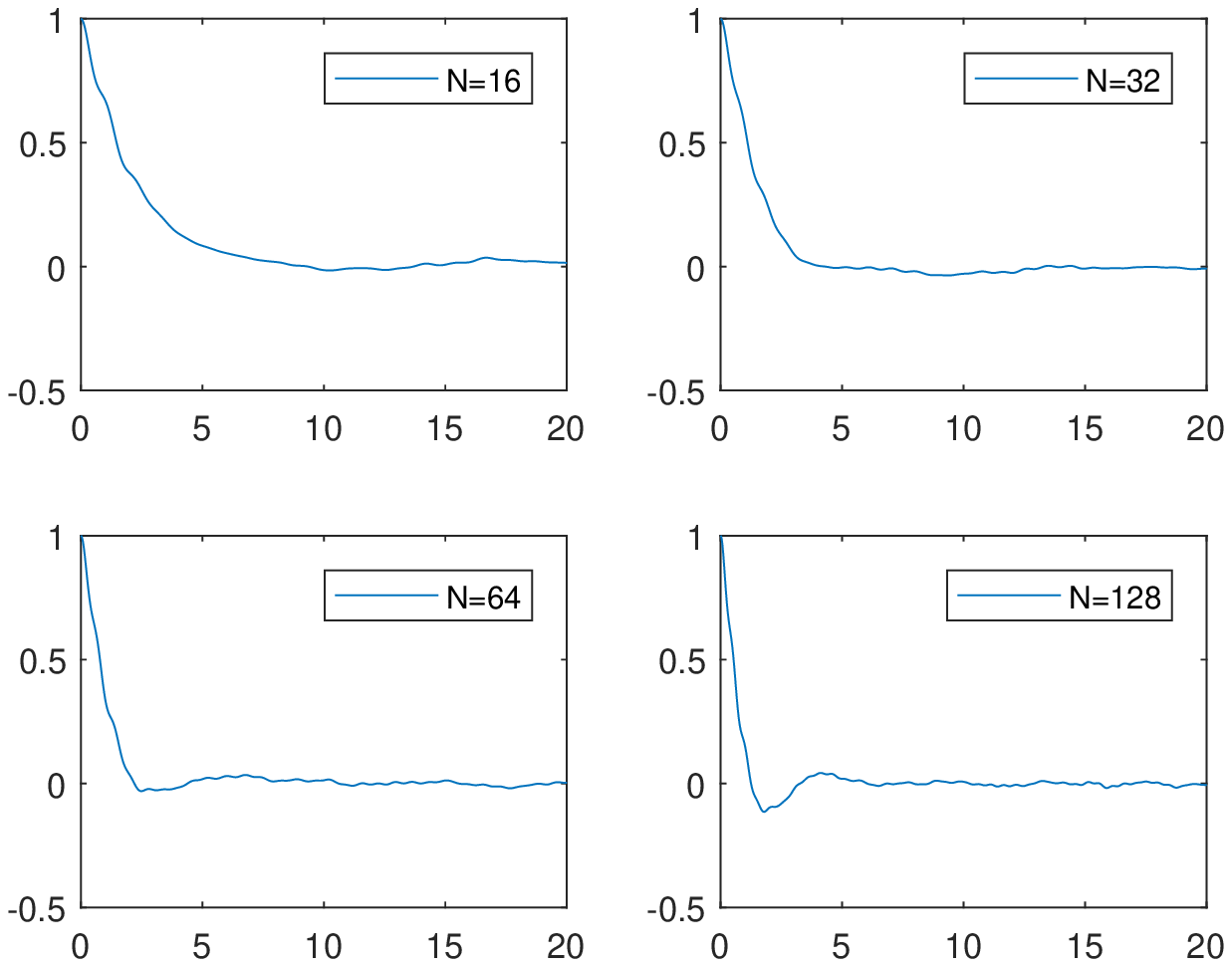}\\
\caption{1D Example.  Autocorrelation of the position variables sampled from the Langevin system in staging coordinates. Left: Staging PIMD with the friction constant $\gamma=1$. Right: Staging PIMD with the setup in \cite{LiuLiLiu:16}.}
\label{fig:qauto_s1}
\end{centering}
\end{figure}

\begin{comment}
\begin{figure}[ht]
\begin{centering}
\includegraphics[scale=0.7]{staging_corr2_new.eps} \\
\caption{1D Example. Staging PIMD with the setup in \cite{LiuLiLiu:16}. Autocorrelation of the position variables sampled from the Langevin system in staging coordinates. %\jl{I suggest merge the two figures 10 and 11 to make a comparison directly and also to save some space}
}
\label{fig:qauto_s2}
\end{centering}
\end{figure}
\end{comment}

\FloatBarrier

\subsection{Tests with the 2D example}

We now consider the 2D test example, with the potential function and
the observable given by \eqref{ex:pot2} and \eqref{ex:ob2},
respectively, and the inverse temperature $\beta=8$.

We only test the preconditioned dynamics \eqref{eq:pLang} and
\eqref{eq:pmmLang}, they have improved numerical stability. To study
the performance of the BAOAB method applied to the two preconditioned
systems, we take the number of beads $N=128$, which makes the
asymptotic error in the ring polymer approximation negligible. We test
the BAOAB scheme for the two dynamics with time step sizes
$\Delta t= \frac 1 2$, $\frac 1 4$, $\frac 1 8$ and $\frac 1 {16}$ and
simulation time $T=40,000$.  We plot the the mean squared errors at
different times in Figure \ref{fig:mse2d}. When
$\Delta t = \frac 1 2$, we observe for either system the sampling
error is saturated around $t=1,000$, and when $\Delta t = \frac 1 4$
or $\frac 1 8$, the sampling error is saturated around
$t=10,000$. This implies that when the simulation time is long enough,
the bias introduced by numerically integrating the sampling
trajectories with large time steps dominates the mean squared error.
Similar to the 1D tests, the numerical results suggest we can take
$o(1)$ time steps for accurate approximation of the ensemble average,
even for very large numbers of beads.

\begin{figure}[ht]
\begin{centering}
\includegraphics[scale=0.5]{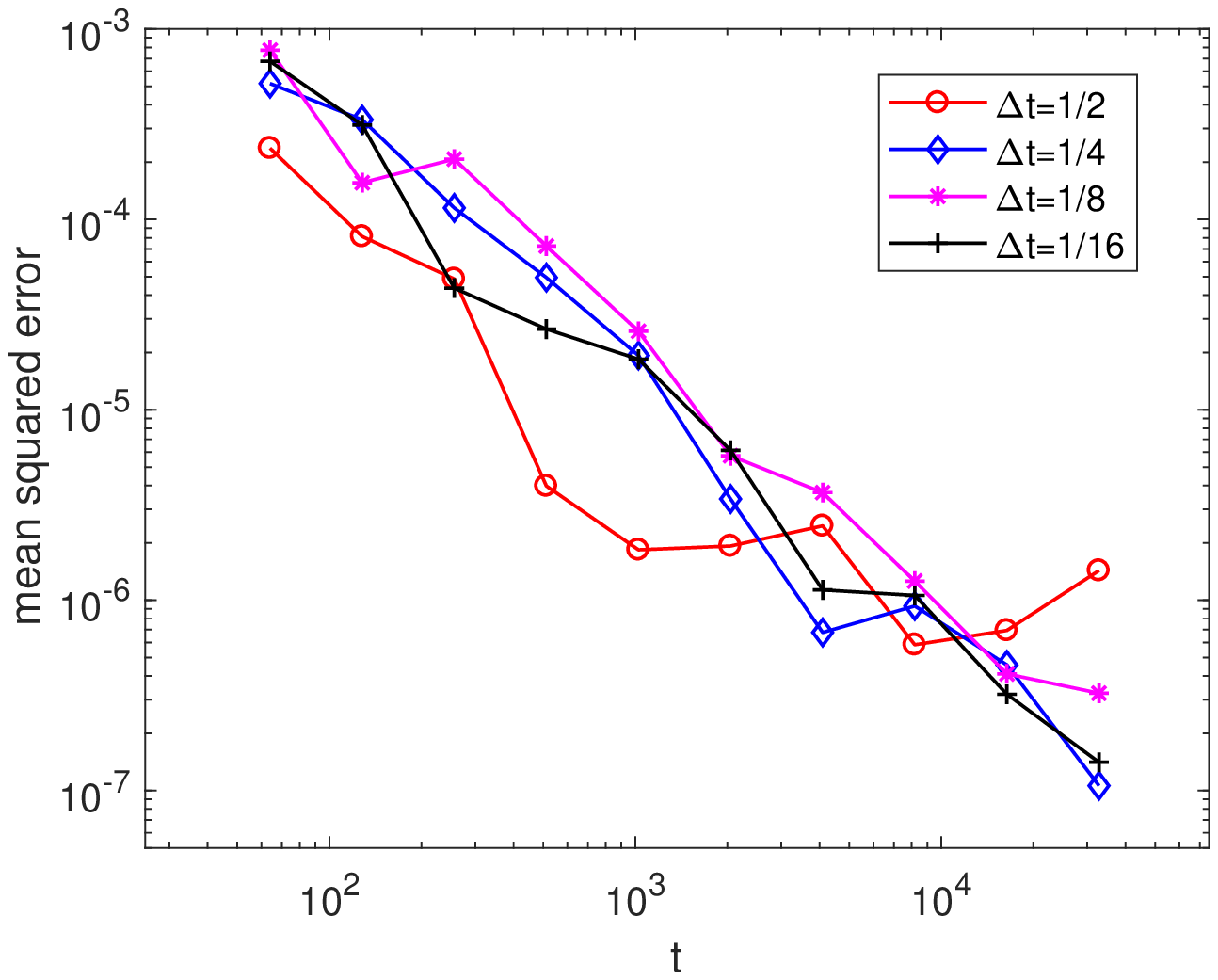} 
\includegraphics[scale=0.5]{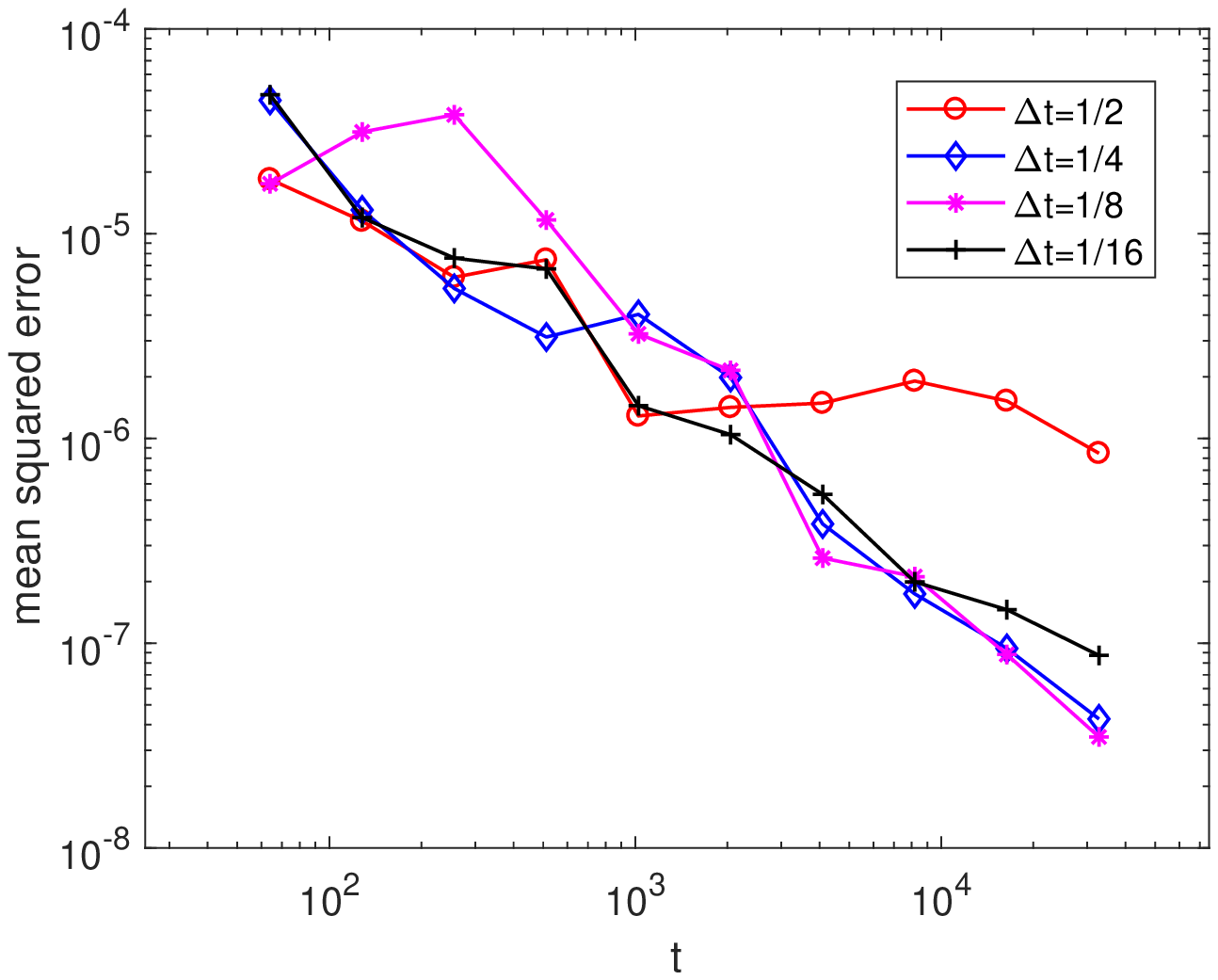} \\
\caption{2D Example. Mean squared errors at
  $t=2^6,\, 2^7, \cdots, 2^{15}$ by the BAOAB method with various time
  step sizes. Left: the \eqref{eq:pLang} dynamics. Right: the
  \eqref{eq:pmmLang} dynamics. %\jl{use dynamics instead of ``system'' for these things} 
}
\label{fig:mse2d}
\end{centering}
\end{figure}

Next, we compare the PIMD simulation for \eqref{eq:pLang} and the
\eqref{eq:pmmLang} for different numbers of beads.  We take
$\Delta t= \frac 1 {50}$ to make sure the bias due to numerical
integration of the sampling trajectory is negligible when the
simulation time $T=10,000$.  The mean squared errors for the two
sampling dynamics are plotted in Figure~\ref{fig:mse2d2}. We observe
in either test noticeable asymptotic error when the number of beads is
$32$ or $64$, while large number of beads reduces the error.

\begin{figure}[ht]
\begin{centering}
\includegraphics[scale=0.5]{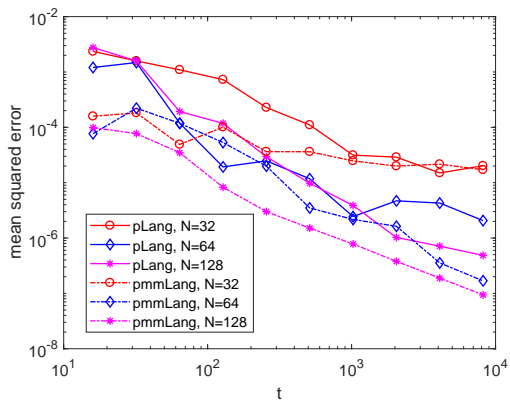} \\
\caption{2D Example. Mean squared errors at $t=2^4,\, 2^5, \cdots, 2^{13}$ by the BAOAB method with various number of beads.  Solid line: the \eqref{eq:pLang} system. Dash-dot line: the \eqref{eq:pmmLang} system.}
\label{fig:mse2d2}
\end{centering}
\end{figure}

Finally, we compare the accuracy of PIMD simulations using \eqref{eq:pLang}  and  \eqref{eq:pmmLang} dynamics. By comparing  Figure~\ref{fig:mse2d} and  Figure~\ref{fig:mse2d2}, we observe that when the numerical error is dominated by the sample variance (small $\Delta t$ and large number of beads), the \eqref{eq:pmmLang} dynamics gives better accuracy in terms of the mean squared error. The autocorrelation time for the $\bd q$ variable plotted in Figure~\ref{fig:qauto2} is similar to the 1D case, which indicates that \eqref{eq:pmmLang} has better sampling efficiency. This is also further confirmed by Table~\ref{table:test2d} which present empirical error, 95\% confidence interval and the mean squared error at simulation time $T=10,000$.

\begin{figure}[ht]
\begin{centering}
\includegraphics[scale=0.5]{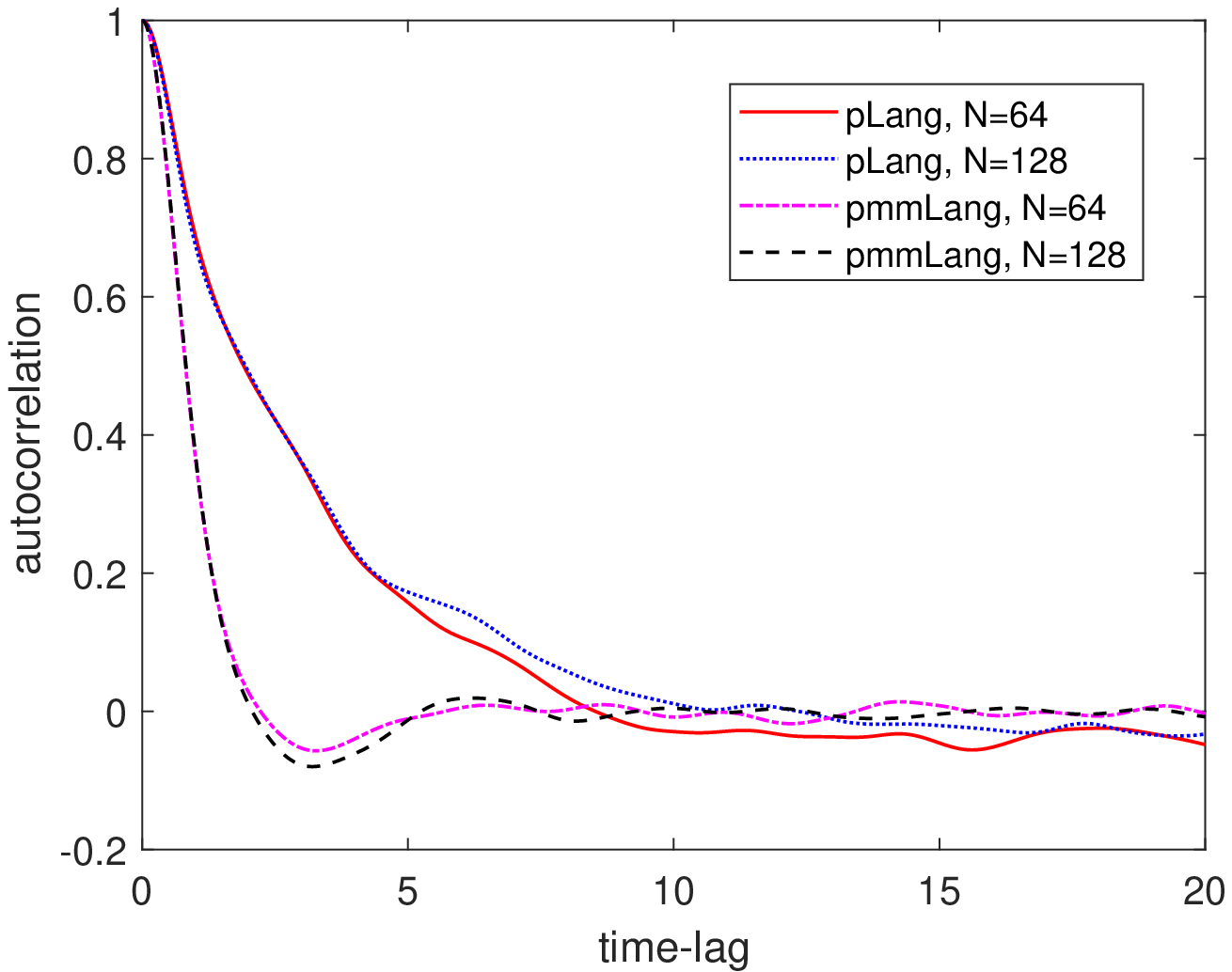} \\
\caption{2D Example. Autocorrelation of the position variables sampled from the \eqref{eq:pLang} system and the \eqref{eq:pmmLang} system.}
\label{fig:qauto2}
\end{centering}
\end{figure}

\begin{table}[htpb]
  \centering
  \begin{tabular}{ c| c| c|c|c} 
\toprule
    & \multicolumn{2}{c|}{\eqref{eq:pLang} dynamics} & \multicolumn{2}{c}{\eqref{eq:pmmLang} dynamics} 
    \\    \hline
    & \,  N=64\, & \, N=128 \, & \, N=64 \, &  \, N=128\,  \\ 
\hline
Error & 1.08e-3 & 5.64e-4 &  1.38e-3 & 7.70e-5 
    \\ \hline
   95\% C.I. & 9.71e-4 & 8.79e-4 & 6.78e-3 & 5.53e-4 
    \\ \hline
M.S.E. & 1.42e-6 & 5.19e-7 & 1.39e-7 & 9.28e-8   \\ \bottomrule
\end{tabular}
\caption{2D Example. Errors in numerical empirical averages with 95 \% confidence intervals and mean squared errors. The reference value is $-0.0888156$. }
  \label{table:test2d}
\end{table}

\subsection{Tests with the two-level system}

%\zz{this section has been slightly reworked}

In this part, we implement the preconditioned sampling dynamics,  pmmIS, for the two-level quantum system \eqref{ex:pot1-2} with two sets of observables \eqref{eq:ob1} and \eqref{eq:ob2}. The numerical results are compared with those sampled by the direct simulation of the PIMS-SH method and its infinite swapping limit, which we abbreviate by ``DS'' and ``IS'', respectively. We choose the inverse temperature $\beta=4$.  We test the BAOAB method with $\Delta t=1$, $\frac 1 4$,
$\frac 1 {16}$, $\frac 1 {64}$ and $\frac 1 {256}$ for the three sampling systems. The errors in the
empirical averages till simulation time $T=10,000$ from each
simulation are reported in Table~\ref{table:test3_1}.

In \cite{LZPIMD1,LZPIMD2}, the authors have shown that the DS method give satisfactory performances for the diagonal observables, but requires smaller time steps in integrating the sampling trajectories for the off-diagonal observables, while numerical simulation based on its infinite swapping limit allows large time steps for all the observables with some reasonable increase in the computational cost. 

The numerical results based on the DS method in Table~\ref{table:test3_1} agree with the previous understanding of the method. Moreover, we observe that, similar to the single level cases, both the DS method and the IS method are unstable for big time steps when the number of beads is large. When the time steps are sufficiently small, the IS simulations show better accuracy for the off-diagonal observables than the DS simulations.

Simulations by pmmIS exhibit improved stability and satisfactory accuracy for the observables, and in particular, the  Table~\ref{table:test3_1} show that when $\Delta t= \frac 1 4$, pmmIS  already gives very good approximations of the thermal averages for both diagonal and off-diagonal observables.   The numerical results clearly verify that the pmmIS simulations  show great accuracy for fairly large time steps. 
%But, the for the off-diagonal observables, although the simulations based on pmmDS are still stable, the approximations are of low quality unless the time steps are sufficiently small. The pmmIS simulations still show great accuracy for fairly large time steps. 

\begin{table}[pthb]
  \centering
  \begin{tabular}{ c| c| c|c|c|c} \toprule
    \# beads  & $\Delta t={1}$ & $\Delta t=\frac{1}{4}$ & $\Delta t=\frac{1}{16}$ & $\Delta t=\frac{1}{64}$   &$\Delta t=\frac{1}{256}$  \\ \hline
 (DS) 32 & NaN & NaN & 8.51e-3 & 6.71e-3 & 4.55e-3
    \\ \hline
 (DS) 64 & NaN & NaN & 5.59e-1 & 6.28e-3 & 4.11e-3 
    \\ \hline
 (IS) 32 & NaN & NaN & 1.39e-3  & 1.18e-3 & 1.38e-3
    \\ \hline
 (IS) 64 & NaN & NaN &  8.39e-1 &   1.66e-3 & 1.67e-4
    \\ \hline
% (pmmDS) 32 & 5.47e-1  & 1.67e-3 & 4.51e-3 & 6.19e-3  & 5.54e-3
%    \\ \hline
% (pmmDS) 64   & 5.49e-1  & 8.06e-3 & 3.76e-3  & 4.26e-3 & 4.25e-3
%    \\ \hline
 (pmmIS) 32  & 2.38e-1  & 2.62e-3 & 3.38e-3 & 3.22e-3  & 1.03e-3
    \\ \hline
 (pmmIS) 64  & 2.33e-1 & 1.22e-3 & 1.80e-3 & 1.63e-3  & 2.20e-5
    \\ \bottomrule
\end{tabular}
\medskip

  \begin{tabular}{ c| c| c|c|c|c} \toprule
    \# beads  & $\Delta t={1}$ & $\Delta t=\frac{1}{4}$ & $\Delta t=\frac{1}{16}$ & $\Delta t=\frac{1}{64}$   &$\Delta t=\frac{1}{256}$  \\ \hline
 (DS) 32 & NaN  & NaN   & 2.93e-2 & 1.52e-2  & 1.17e-2
    \\ \hline
 (DS) 64 & NaN &  NaN & 6.78e-1 & 2.62e-2  &  3.29e-2
    \\ \hline
 (IS) 32 & NaN & NaN & 2.46e-3  & 1.62e-3 &  6.95e-4
    \\ \hline
 (IS) 64 & NaN & NaN & 8.16e-1  & 5.01e-4  & 1.09e-3
    \\ \hline
% (pmmDS) 32  & 6.42e-1 & 4.02e-1 & 3.67e-2  & 6.52e-3 & 8.29e-3
%    \\ \hline
% (pmmDS) 64  & 1.53e0  & 1.12e-0 & 2.53e-2 &  1.91e-2 & 1.35e-2
%    \\ \hline
 (pmmIS) 32 & 3.07e-1  & 2.71e-3 &  2.51e-3 &  1.12e-3 &  2.02e-3
    \\ \hline
 (pmmIS) 64  & 3.10e-1 & 6.62e-3 & 4.12e-3  & 4.72e-4 & 9.90e-4
    \\ \bottomrule
\end{tabular}
\medskip

\caption{Two-level Example. Numerical empirical averages computed with various time step sizes and various numbers of beads.  Top: the diagonal observable. The reference value is 8.2234e-1. Bottom: the off-diagonal observable. The reference value is -8.1785e-1. ``NaN'' means the numerical integrator is unstable.}
  \label{table:test3_1}
\end{table}

\section{Conclusion}

We have introduced two preconditioned Langevin sampling dynamics for
path-integral molecular dynamics, which are in particular effective
when the number of beads is large. The forcing from the stiff spring
potential between beads is replaced by a linear damping term, which
allows large time steps for numerically integration. In terms of the
normal modes representation, the mapped modes in the preconditioned
Langevin approach has a uniform upper bound in frequency while the
mapped slow modes may take longer time for sampling. The mapped modes
in the preconditioned mass-modified Langevin sampling all have the
same frequency $1$, and the corresponding Langevin dynamics has a
natural connection to the continuum limit as the number of beads goes
to infinity. The numerical tests validate the improved stability and
better sampling accuracy for both preconditioned Langevin sampling
dynamics for thermal averages. 

%For future works, it is of interest to
%extend the continuum limit and preconditioning schemes for
%path-integral molecular dynamics with surface hopping
%\cite{LZPIMD1, LZPIMD2} for non-adiabatic quantum systems.

\section*{Acknowledgments}
  The work of J.L.~is partially supported by the National Science
  Foundation under grant DMS-1454939. The work of Z.Z.~is partially
  supported by a start-up fund from Peking University and NSFC
  11801016. We thank Jian Liu for useful
  discussions.  
%\end{acknowledgments}

\appendix
%\appendixpage

\section{The covariance operator and its finite dimensional approximation}  \label{sec:cov}

%\jl{is this appendix ever referred to in the main text?} 

The precondition schemes we proposed rely on the covariance operator
$\mathfrak C^\alpha =(\mathfrak L^\alpha)^{-1}$ and its finite
dimensional approximation $(L^\alpha)^{-1}$.  In this appendix,
we % provide an alternative way to directly define the covariance operator $\mathfrak C^\alpha$and show that it coincides with the inverse of the operator $\mathfrak L^\alpha$. Also,
we derive the explicit expression for $\mathfrak C^\alpha$, and
discuss its discretizations.

To explicitly calculate the covariance operator $\mathfrak C^\alpha$, we solve for the covariance function $C^\alpha(\tau,\tau')$ such that the covariance operator is the integral operator with $C^\alpha(\tau,\tau')$ as its kernel, i.e.
\begin{equation}\label{eq:covf}
\mathfrak C^\alpha f (\tau) = \int_0^\beta C^\alpha (\tau,\tau') f(\tau') \ud \tau', \quad \tau \in [0, \beta].
\end{equation}
The covariance function satisfies the following boundary value problem, for $\tau,\tau'\in [0, \beta]$,
\begin{align} 
\mathfrak L^\alpha C^\alpha(\tau,\tau')&= \delta(\tau-\tau'); \label{eq:bvp1} \\
C^\alpha(0,\tau')&=C^\alpha(\beta,\tau'), \label{eq:bvp2} \\
C_{\tau}^\alpha(0,\tau')&=C_{\tau}^\alpha(\beta,\tau').  \label{eq:bvp3}
\end{align}

We show in the following the covariance operator as in \eqref{eq:covf} indeed gives the inverse of the $\mathfrak L^\alpha$ operator. For $f(\tau)$ and $g(\tau)$ satisfying the periodic boundary conditions, by Green's formula, we can easily show that
\[
\int_0^{\beta} f (\tau) \mathfrak L^\alpha g (\tau) - g(\tau) \mathfrak L^\alpha f (\tau) \ud \tau =0. 
\] 
Now we take $g(\tau)= C^\alpha(\tau,\tau')$, then $\mathfrak L^\alpha g (\tau)=  \delta (\tau-\tau')$, then we have
\[
f(\tau') - \int_0^{\beta} C^\alpha(\tau,\tau') \mathfrak L^\alpha f (\tau) \ud \tau =0,
\] 
Since it can be shown that $C^\alpha(\tau,\tau')$ is symmetric, we thus obtain
\[
f=\mathfrak C^\alpha \mathfrak L^\alpha f.
\]
This verifies $\mathfrak C^\alpha = (\mathfrak L^\alpha)^{-1}$.

Let us write down the explicit expression when $d = 1$ and
$\alpha = 1$. The extension to general cases is straightforward.  By
solving the boundary value problem \eqref{eq:bvp1}--\eqref{eq:bvp3},
we get
\begin{equation}
  C^1(\tau,\tau')=
  \begin{cases}
\frac{e^{\tau'-\tau}}{2(e^\beta-1)}+ \frac{e^{\tau-\tau'}}{2(1-e^{-\beta})}, & 0<\tau<\tau'<\beta, \\
\frac{e^{\tau'-\tau}}{2(1-e^{-\beta})}+\frac{e^{\tau-\tau'}}{2(e^\beta-1)}, & 0<\tau'<\tau<\beta.
\end{cases}
\end{equation}
In a more compact form, we can write
\[
C^1(\tau,\tau')=\frac{e^{|\tau-\tau'|}}{2(e^\beta-1)}+\frac{e^{-|\tau-\tau'|}}{2(1-e^{-\beta})}.
\]
From this, we see clearly, the covariance function is symmetric and is a function of $|\tau-\tau'|$.

Next, we discuss  two types of finite dimensional approximation  of the covariance operator, where the first one is based on the analytical expression of the covariance function and the second one is based on the inverse of the finite dimensional approximation of $\mathfrak L^\alpha$. 

If we denote the equidistant grid points in $\tau$ by $\tau=s_i$, $i=1,\cdots,N$, and evaluate the covariance function at those grid points, we obtain the numerical approximation of the covariance operator $\mathfrak C^1$, which is denoted by $C^1_N$. Namely, given a function $g: [0,\beta] \mapsto \R$ and we denote its confinement on the grids $\{ \tau_i \}$  by $\bd g$,  then we have
\[
C^1_N \, \bd g = \sum_{j=1}^N  C^1(\tau_i,\tau_j) g (\tau_j) \beta_N. 
\]
With a bit abuse of notations, we can also view $C^1_N $ as a matrix, such that 
\[
\left(C^1_N \right)_{ij}=  C^1(\tau_i,\tau_j) \beta_N. 
\]
However, this approximate covariance is not exactly the inverse of $L^1$ as in \eqref{eq:La} on the same grids, since $L^\alpha$ is only a finite difference approximation to the continuous counterpart as in \eqref{eq:cLa}. 

An alternative way is to directly take the inverse of the finite
dimensional approximation of $\mathfrak L^1$ on the grid
points. Consider the equidistant grid points $\{ \tau_i \}$, we
observe the finite difference approximation of the covariance operator
$\mathfrak L^1$ is exactly $L^1$ (viewed as a linear transform)
defined as in \eqref{eq:La}. Clearly, $L^1$ is strictly positive
definite, and is thus invertible. We will use this approach to
precondition finite dimensional systems. 
%Hence, we have
%%
%%If we denote the grid points in $s$ by $s=s_i$, $i=1,\cdots,N$, and denote the numerical approximation of $-\partial_{ss}$ with periodic boundary conditions by $\widetilde L$, which might be obtained by finite difference approximation or pseudo-spectral approximation, and as before we denote by $I$ the identity matrix.  Then, then numerical approximation of the operator $\mathfrak L^1$ , $\widetilde L^1$, is given by
%%\[
%%\widetilde L^1 = \widetilde L+I.
%%\] 
%%Thus, another numerical approximation of the covariance operator $\mathfrak C$ is obtained, which is denoted by $C^b$, namely
%\begin{equation} \label{eq:cnb}
%C_b:=( L+I)^{-1}= (L^1)^{-1}.
%\end{equation}
%Note that, $C_b$ can be viewed as a matrix as well as linear operator. 

% We expect that when $N \gg 1$, $(L^1)^{-1}$  is close to $C^1_N$, which can be explicitly solved. 
% We remark that, since the continuum limit is identified, we can use other ways to approximate  the differential operator $\mathfrak L^\alpha$, and thus derive different preconditioning matrices.
% For the rest of this paper, we only consider $(L^\alpha)^{-1}$ as the preconditioning matrix in finite dimensions.

\section{Invariance measure in finite dimensional cases}  \label{app:inv}

In this section, we aim to verify that the finite dimensional Langevin dynamics, \eqref{eq:Lang}, \eqref{eq:pLang} and \eqref{eq:mmLang} all take $\pi_N(\bd q, \bd p)$ as their invariant measures (although the choices of the mass matrix are different), and \eqref{eq:pmmLang} takes the invariance measure $\widetilde \pi_N (\bd q, \bd v)$. 

Note that, albeit various choices of the mass matrices, in general  $\pi_N(\bd q, \bd p)$ is given by
\begin{equation} \label{eq:pi1}
\pi_N (\bd q, \bd p) \propto \exp \left(- \beta_N  \left( \frac 1 2 \bd q \cdot L^\alpha \bd q +  U_N^\alpha+ \frac{1}{2} \bd p \cdot M^{-1} \bd p   \right) \right).
\end{equation}

We consider the Langevin dynamics as defined in \eqref{eq:qagen} and \eqref{eq:pagen}, which covers the systems \eqref{eq:Lang}, \eqref{eq:pLang} and \eqref{eq:mmLang}.
%\begin{align} \label{eq:qagen}
%  {\ud \bd q}&=  {C_1 M^{-1} \bd p} \ud t;\\
%  {\ud  \bd p}&= -C_1 L^\alpha \bd q \ud t -C_1\nabla_{\bd q} U^\alpha_N \ud t   \label{eq:pagen} \\
%& \quad \quad - \gamma  C_2 \bd p\ud t + \sqrt{\frac{2 \gamma   C_2 M}{\beta_N}} \ud \bd B,  \nonumber
%\end{align}
%where $C_1$ and $C_2$ are some positive definite matrices, and we assume $C_2$ and $M$ are commutable, $C_2M=MC_2$. %And we call \eqref{eq:qagen} and \eqref{eq:pagen} the generic Langevin dynamics (GL). 
%Obviously, when $C_1=C_2=I$ and $M= m I$, the system \eqref{eq:qagen}  \eqref{eq:pagen} reduces to the \eqref{eq:Lang} system. When $C_1=C_2=(L^\alpha)^{-1}$ and $M= m I$,  the system \eqref{eq:qagen}  \eqref{eq:pagen}  reduces to pLan. And finally, when $C_1=C_2=I$ and $M= L^\alpha$,  it  reduces to the \eqref{eq:mmLang} system.
%
%The Fokker-Planck equation corresponding to \eqref{eq:qa} and \eqref{eq:pa} reads
%\begin{equation}\label{eq:fk1}
%\frac{\partial }{\partial t} f+ \frac{\bd p}{m} \cdot \nabla_{\bd q} f   - \left(-D^{(2)} \bd q+\nabla_{\bd q} V_N \right)\cdot {\nabla_\bd p} f =   \gamma {\nabla_\bd p} \cdot \left( \bd p f + \frac{m}{\beta_N} {\nabla_\bd p} f  \right).
%\end{equation}
%Clearly, the phase space distribution \eqref{eq:pi1} is the steady state solution to \eqref{eq:fk1}, and it can be shown that the Langevin dynamics \eqref{eq:qa} and \eqref{eq:pa} takes  \eqref{eq:pi1} as its invariant measure.

The Fokker-Planck equation corresponding to \eqref{eq:qagen} and \eqref{eq:pagen} reads
\begin{equation} \label{eq:fk1pre}
\frac{\partial }{\partial t} f+  C_1  M^{-1} {\bd p} \cdot \nabla_{\bd q} f   + \left( (-C_1 L) \bd q- C_1 \nabla U_N \right)\cdot \nabla_{\bd p} f 
=   \gamma \nabla_{\bd p} \cdot \left( C_2 \bigl(\bd p f + \frac{M}{\beta_N} \nabla_{\bd p} f \bigr)  \right).
\end{equation}
Here, we have used the fact that
\begin{equation*}
\nabla_{\bd p}^2 : (C_2 M) = \sum_{ij} \partial_{p_i} \partial_{p_j} (C_2 M)_{ji} \\
 = \sum_{ij}\partial_{p_i} (C_2 M)_{ji}  \partial_{p_j} = \nabla_{\bd p} \cdot (C_2 M \nabla_{\bd p}).
\end{equation*}

Hence, we can easily see that,  \eqref{eq:pi1} is  an steady state to this Fokker-Planck equation. Therefore, \eqref{eq:Lang}, \eqref{eq:pLang} and \eqref{eq:mmLang} all have invariant measures as in \eqref{eq:pi1}. We also conclude that the preconditioning with $C_1$ and $C_2$ does not change the invariant measure. 

In particular, in \eqref{eq:pLang},  the Fokker-Planck equation takes the following form 
\begin{equation*}
\frac{\partial }{\partial t} f+  (L^\alpha)^{-1}  \frac 1 m {\bd p} \cdot \nabla_{\bd q} f   + \left( - \bd q- (L^\alpha)^{-1} \nabla U_N \right)\cdot \nabla_{\bd p} f  
=   \gamma \nabla_{\bd p} \cdot \left( (L^\alpha)^{-1} \bigl(\bd p f + \frac{m}{\beta_N} \nabla_{\bd p} f \bigr)  \right).
\end{equation*}
We observe that, the most stiff part of the original equation $-L^\alpha \bd q$ is replaced by $- \bd q$ due to the preconditioning. 

%For ULB, we have chosen $M=G$, and 
%\begin{equation} \label{eq:pi2}
%\pi_N (\bd q, \bd p) \propto \exp \left(-   \left( \frac  1 2 \bd p \cdot G^{-1} \bd p + \frac  1 2 \bd q \cdot G \bd q +  \beta_N U_N  \right) \right),
%\end{equation}
For the \eqref{eq:pmmLang} system, we have the phase space distribution in $(\bd q,\,\bd v)$ variables,
\begin{equation} \label{eq:pi3}
\widetilde \pi_N (\bd q, \bd v) \propto =\exp \left(-   \beta_N \left( \frac  1 2 \bd v \cdot L^\alpha \bd v + \frac  1 2 \bd q \cdot L \bd q +   U^\alpha_N  \right) \right).
\end{equation}

%The Fokker-Planck equation corresponding to \eqref{eq:qaM} and \eqref{eq:paM} reads
%\begin{equation}\label{eq:fk1M}
%\frac{\partial }{\partial t} f+ G^{-1}{\bd p} \cdot \nabla_{\bd q} f   - \left( G \bd q+\nabla_{\bd q} \beta_N U_N \right)\cdot \nabla_{\bd p} f =   \gamma \nabla_{\bd p} \cdot \left( \bd p f + G \nabla_{\bd p} f  \right).
%\end{equation}
%Clearly, it takes \eqref{eq:pi2} as its invariant measure.
Clearly, the Fokker-Planck equation corresponding to \eqref{eq:pmmLang} as in \eqref{eq:pmmLang} is given by
\begin{equation}\label{eq:fk1M}
\frac{\partial }{\partial t} f+ {\bd v} \cdot \nabla_{\bd q} f   - \left(  \bd q+(L^\alpha)^{-1} \nabla_{\bd q} U_N \right)\cdot \nabla_{\bd v} f  
=   \gamma \nabla_{\bd p} \cdot \left( \bd p f +  \frac{(L^\alpha)^{-1}}{\beta_N} \nabla_{\bd p} f  \right).
\end{equation} 
We can rewrite this equation as
\begin{equation*}
\frac{\partial }{\partial t} f+ {(L^\alpha)^{-1} L^{\alpha}\bd v} \cdot \nabla_{\bd q} f   - (L^\alpha)^{-1} \left(  L^\alpha \bd q+ \nabla_{\bd q} U_N \right)\cdot \nabla_{\bd v} f \\
 =   \gamma \nabla_{\bd p} \cdot \left( \bd p f +  \frac{(L^\alpha)^{-1}}{\beta_N} \nabla_{\bd p} f  \right).
\end{equation*}
Therefore, we can verify that the \eqref{eq:pmmLang} system
\eqref{eq:pmmLang} takes \eqref{eq:pi3} as its invariant
measure. Finally, we observe that, similar to the Fokker-Planck
equation for \eqref{eq:pLang}, \eqref{eq:fk1M} also includes the term
$-\bd q$ in place of the stiff term $-L^\alpha \bd q$, but the rest of
the terms are different.

\section{Ring polymer representation for  two-level quantum systems} \label{app:two-level}

In this part, we present the details of the ring polymer representation of thermal averages as in \eqref{eq:aveA}  for two-level quantum systems, which have been rigorously derived in \cite{LZPIMD1}. With the diabatic basis, we approximate the partition
function by a ring polymer representation with $N$ beads
\[
\label{eq:reZ} \tr_{ne}[e^{-\beta \wh H}] \approx
\mc{Z}_N  := \frac{1}{(2\pi)^{dN}} \int_{\RR^{2dN}} \ud \bd q \ud \bd p 
\times \sum_{\bd{\ell} \in \{0, 1\}^N}  \exp(-\beta_N H_N(\bd{q},
\bd{p}, \bd{\ell})),
\] 
where $\beta_N=\beta/N$. The ring polymer that consists
of $N$ beads is prescribed by the configuration
$(\bd{q}, \bd{p}, \bd{\ell}) \in \RR^{dN} \times \RR^{dN} \times \{0,
1\}^N$.

For a given ring
polymer with configuration $(\bd{q}, \bd{p}, \bd{\ell})$, the
effective Hamiltonian $H_N(\bd{q}, \bd{p}, \bd{\ell})$ is given by
\begin{equation}\label{eq:Ham} H_N(\bd q, \bd p, \bd\ell) =  \frac 1 2 \bd p \cdot  M^{-1} \bd p +
\sum_{k=1}^N \bra{\ell_k} S_k \ket{\ell_{k+1}},
\end{equation} 
where we take the convention that $\ell_{N+1} = \ell_1$
and matrix elements of $S_k$, $k = 1, \ldots, N$, are given by
\begin{subequations} \label{eq:Gele}
\[
  \bra{\ell} S_{k} \ket{\ell'} =  \frac{ M\left(q_k-q_{k+1}
\right)^2 }{2(\beta_N)^2}  +\frac{V_{00}(q_k)+V_{11}(q_k)}{2} 
-\frac{1}{\beta_N}\ln\Bigl( \sinh \bigl(\beta_N \Abs{V_{01}(q_k)}
\bigr) \Bigr), 
\]
for $\ell \ne \ell'$, and the diagonal terms are given as
\[ \bra{\ell} S_{k} \ket{\ell} = 
  \frac{M \left(q_k-q_{k+1} \right)^2}{2(\beta_N)^2}  +V_{\ell\ell}(q_k)
  - \frac{1}{\beta_N} \ln \Bigl( \cosh \bigl(
  \beta_N \Abs{ V_{01} (q_k) }\bigr) \Bigr),
\]
\end{subequations}
where we have suppressed the $\bd{q}$ and $\bd{p}$ dependence in the
notation of $S_k$.  Here $S_k$ can be understood as the the contribution of $\bra{q_k} e^{-\beta_N \wh H} \ket{q_{k+1}}$
to the effective Hamiltonian $H_N$ in the ring polymer representation. The readers may refer to \cite{LZPIMD1} for the derivations.

For an observable $\wh{A}$, under the ring polymer representation, we
have
\begin{equation} \label{eq:reAo} \tr_{ne}[e^{-\beta \wh H}\wh A ]
\approx \frac{1}{(2\pi)^{dN}} \int_{\RR^{2dN}} \ud \bd q \ud \bd p
\sum_{\bd{l} \in \{0, 1\}^N} \\ \times \exp(-\beta_N H_N) W_N[A],
\end{equation} where the weight function associated to the observable
is given by (recall that $\wh{A}$ only depends on position by our
assumption)
\begin{equation}\label{eq:WNA} W_N[A] (\bd q, \bd p, \bd\ell) =
  \frac{1}{N} \sum_{k=1}^N \biggl( \bra{\ell_k} A(q_k) \ket{\ell_k} \\
  - e^{\beta_N \langle \ell_k| S_k| \ell_{k+1} \rangle - \beta_N
    \bra{\bar{\ell}_k} S_{k} \ket{ \ell_{k+1}}} \bra{\ell_{k}}
  A(q_{k}) \ket{\bar{\ell}_k}
  \frac{V_{\ell_k\bar{\ell}_k}}{\abs{V_{\ell_k\bar{\ell}_k}}} \biggr),
\end{equation}
where we have introduced the short hand notation
$\bar{\ell}_k = 1 - \ell_k$, \textit{i.e.}, $\bar{\ell}_k$ is the
level index of the other potential energy surface than the one
corresponds to $\ell_k$ in our two-level case.  Similar as for the
partition function, the ring polymer representation \eqref{eq:reAo}
replaces the quantum thermal average by an average over ring polymer
configurations on the extended phase space
$\R^{dN} \times \R^{dN} \times \{ 0, 1 \}^N$.

The ring polymer representation for a multi-level quantum system can
be also constructed using the adiabatic basis
\cite{SchmidtTully2007,LZPIMD1}, and much of the current work
also extends to the ring polymer with the adiabatic basis. We will skip
the details and leave to interested readers.

\bibliographystyle{plain}

\bibliography{surfacehoppingPIMDcontconv} 

\end{document}